\documentclass[article, a4]{amsart}
\usepackage {amsmath, amssymb, color, textcomp}
\usepackage{mathtools}
\DeclarePairedDelimiter{\floor}{\lfloor}{\rfloor}

\newlength{\extramargin}

\makeatletter
\newtheorem{thm}{Theorem}[section]
\newtheorem{cor}[thm]{Corollary}
\newtheorem{lem}[thm]{Lemma}
\newtheorem{defn}[thm]{Definition}
\newtheorem{prop}[thm]{Proposition}

\newtheorem{preremark}[thm]{Remark}
\newenvironment{remark}%
  {\begin{preremark}\upshape}{\end{preremark}}
\newtheorem{preexample}[thm]{Example}
  {\begin{preexample}\upshape}{\end{preexample}}

\numberwithin{equation}{section}

\numberwithin{equation}{section}

\numberwithin{thm}{section}
\newtheorem{lem/defn}[thm]{Lemma/Definition}
\newtheorem{preex/defn}[thm]{Example/Definition}
\newenvironment{ex/defn}%
  {\begin{preex/defn}\upshape}{\end{preex/defn}}

\newcommand{\ten}{\otimes}

\newcommand{\abs}[1]{\lvert#1\rvert}

\DeclareMathOperator{\End}{End}

\begin{document}

\title[Multi-local Virasoro representations on the Fock space $\mathit{F^{\otimes \frac{1}{2}}}$]{Boson-fermion correspondence of type D-A and multi-local Virasoro representations on the Fock space $\mathit{F^{\otimes \frac{1}{2}}}$}
\author{Iana I. Anguelova}

\address{Department of Mathematics,  College of Charleston,
Charleston SC 29424 }
\email{anguelovai@cofc.edu}

\date{\today}

\begin{abstract}
We construct the bosonization of the Fock space $\mathit{F^{\otimes \frac{1}{2}}}$ of a single neutral fermion by using a 2-point local Heisenberg field. We  decompose the Fock space $\mathit{F^{\otimes \frac{1}{2}}}$ as a direct sum of irreducible highest weight  modules for the Heisenberg algebra $\mathcal{H}_{\mathbb{Z}}$, and thus we show that under the  Heisenberg  $\mathcal{H}_{\mathbb{Z}}$ action the Fock space  $\mathit{F^{\otimes \frac{1}{2}}}$ of the single neutral fermion is isomorphic to the Fock space $\mathit{F^{\otimes 1}}$ of a pair of charged free fermions, thereby constructing the boson-fermion correspondence of type D-A.  As a corollary we obtain the Jacobi identity equating the graded dimension formulas utilizing both the Heisenberg and the Virasoro gradings on $\mathit{F^{\otimes \frac{1}{2}}}$.
We construct  a family of   2-point-local   Virasoro fields with central charge $-2+12\lambda -12\lambda^2, \ \lambda\in \mathbb{C}$,  on the Fock space $\mathit{F^{\otimes \frac{1}{2}}}$. We construct a $W_{1+\infty}$ representation on $\mathit{F^{\otimes \frac{1}{2}}}$ and   show that  under the $W_{1+\infty}$ action   $\mathit{F^{\otimes \frac{1}{2}}}$  is again isomorphic to  $\mathit{F^{\otimes 1}}$.
\end{abstract}

\maketitle

\tableofcontents

\section{Introduction}
\label{sec:intro}
This paper is part of a series   studying various particle correspondences in 2 dimensions from the point of view of chiral algebras (vertex algebras) and representation theory.
There are several types of particle correspondences, such as the boson-fermion and boson-boson correspondences, known in the literature. The best known  is the charged free fermion-boson correspondence, also known as  the
boson-fermion correspondence of type A (the name "type A" is due to the fact that this correspondence is canonically related  to the basic representations of the Kac-Moody algebras of type A, see \cite{DJKM-1},  \cite{Frenkel-BF}, \cite{KacRaina}). The correspondence of type A is an isomorphism of super vertex
algebras, but most boson-fermion correspondences cannot be described by the
concept of a super vertex algebra due to the more general singularities in their operator
product expansions. In order to describe the more general cases, including the correspondences of types B, C and D-A, in \cite{Ang-Varna2} and \cite{AngTVA} we defined the  concept of a twisted vertex algebra which generalizes super vertex algebra. By utilizing the bicharacter construction in \cite{AngTVA} we showed that the correspondences of types B, C  and D-A are isomorphisms of twisted vertex algebras.
There have been a long-going research on the different boson-fermion correspondences and their connections with the representation theory of various infinite-dimensional Lie algebras.
The first  to write on the topic of the relationship between a boson-fermion correspondence and representation theory was Igor Frenkel in \cite{Frenkel-BF} and Date, Jimbo, Kashiwara, Miwa in \cite{DJKM3}, \cite{DJKM-1}. Since then many attempts have been made to understand the boson-fermion correspondences as nothing more but an isomorphism of infinite-dimensional Lie algebra modules, i.e., to understand  a boson-fermion correspondence as nothing more but an isomorphism between two realizations (one bosonic and one fermionic) of certain basic modules for a given infinite Lie algebra (such as $a_{\infty}=\hat{gl}_{\infty}$ or $\widehat{sl}_n$, or $\widehat{\mathcal{D}}=W_{1+\infty}$). Starting with Igor Frenkel's work in \cite{Frenkel-BF}, and later, the boson-fermion correspondence of type A  was related to different kinds of Howe-type dualities (see e.g. \cite{WangKac}, \cite{WangDuality}, \cite{WangDual}), and again there was an attempt to understand a boson-fermion correspondence as being identified by, or equivalent to,  a certain duality or a basic  modules decomposition (see e.g. \cite{WangDuality}, \cite{WangDual}).
But what we contend is that the boson-fermion correspondences are more than just maps (or isomorphisms) between certain Lie algebra modules : a boson-fermion correspondence is first and foremost an isomorphism between two different   chiral field theories, one fermionic (expressible in terms of free fermions and their descendants), the other bosonic (expressible in terms  of exponentiated bosons). What we  show is that although the isomorphism of Lie algebra representations (and indeed various dualities) follow
from a boson-fermion correspondence, the isomorphism as Lie algebra modules is  \textbf{not equivalent} to a
boson-fermion correspondence. In fact,  as Igor Frenkel was careful to summarize in Theorem II.4.1 of his very influential paper \cite{Frenkel-BF}, "the canonical"
isomorphism of the two $o(2l)$-current algebra representations in a bosonic and the fermionic Fock spaces \textbf{follows} from the boson-fermion correspondence (in fact that is what makes the isomorphism canonical), but not vice versa.
 In particular, we consider the \textbf{single} neutral fermion Fock space  $\mathit{F^{\ten \frac{1}{2}}}$ and show that as modules for various  Lie algebras  $\mathit{F^{\ten \frac{1}{2}}}\cong \mathit{F^{\ten 1}}$ ($\mathit{F^{\ten 1}}$ is the   Fock space for \textbf{a pair} of free charged fermions, see \cite{KacRaina}, \cite{Kac}, and Remark \ref{remark:F^1} for some details, and is the fermionic side of the boson-fermion correspondence of type A). In particular, in \cite{ACJ2}, we showed that $\mathit{F^{\ten \frac{1}{2}}}$ carry a representation of the $a_{\infty}$ algebra, and as $a_{\infty}$ modules, $\mathit{F^{\ten \frac{1}{2}}}\cong \mathit{F^{\ten 1}}$. In this paper we consider the Heisenberg algebra $\mathcal{H}_{\mathbb{Z}}$ most often associated with bosonization and boson-fermion correspondences (the corresponding Heisenberg field is the free boson field). By  constructing the free boson Heisenberg field on $\mathit{F^{\ten \frac{1}{2}}}$ from  the neutral fermion generating field we show that $\mathit{F^{\ten \frac{1}{2}}}$ carry a representation of the Heisenberg algebra $\mathcal{H}_{\mathbb{Z}}$. Further, we decompose the Fock space $\mathit{F^{\otimes \frac{1}{2}}}$ as a direct sum of irreducible highest weight  modules for the Heisenberg algebra $\mathcal{H}_{\mathbb{Z}}$. This decomposition shows that as  Heisenberg  $\mathcal{H}_{\mathbb{Z}}$ modules the Fock space  $\mathit{F^{\ten \frac{1}{2}}}$ of the single neutral fermion is isomorphic to the Fock space $\mathit{F^{\ten 1}}$ of a  pair of free charged fermions. This is one of the reasons we call this bosonization  "boson-fermion correspondence of type D-A", as $\mathit{F^{\ten \frac{1}{2}}}$ is canonically related  to the basic representations of the Kac-Moody algebras of type D (see e.g. \cite{WangKac}, \cite{WangDuality}, \cite{WangDual}, see also \cite{AngTVA}, \cite{ACJ}).

As expected in chiral conformal field theory, many examples of super vertex algebras have a  Virasoro structure. Super vertex algebras with a Virasoro field  are called vertex operator algebras (see e.g. \cite{FLM}, \cite{FHL}, \cite{LiLep}, subject also to additional axioms), or conformal vertex algebras (\cite{Kac}, \cite{FZvi}); and are extensively studied. In particular, the boson-fermion correspondence of type A has a  family of Virasoro fields, $L^{A, \lambda}(z)$, parametrized by $\lambda \in \mathbb{C}$, with central charge $-12\lambda^2 +12\lambda -2$   (see e.g. \cite{KacRaina}, \cite{Kac}).
 As was  done for super vertex algebras, in a series of papers we study the existence of Virasoro fields (see Definition \ref{defn:VirStr}) in important examples of twisted vertex algebras, such as the  correspondences of types B, C and D-A. In the case of $\mathit{F^{\ten \frac{1}{2}}}$ (type D-A), it is well known that as part of its conformal super vertex algebra structure  $\mathit{F^{\ten \frac{1}{2}}}$ carries a representation of the Virasoro algebra with central charge $\frac{1}{2}$ (see e.g. \cite{Triality}, \cite{WangKac}).  Using this well-know Virasoro structure, and the new Heisenberg structure on $\mathit{F^{\otimes \frac{1}{2}}}$
 we obtain the Jacobi identity equating two graded dimension formulas representing the Heisenberg and Virasoro gradings on $\mathit{F^{\otimes \frac{1}{2}}}$.

 Next, we use the 2-point locality to construct a Virasoro field with central charge 1 on $\mathit{F^{\ten \frac{1}{2}}}$, and indeed a 2-parameter family of 2-point Virasoro fields on $\mathit{F^{\ten \frac{1}{2}}}$ (with 1-parameter central charge $-12\lambda^2 +12\lambda -2$), see Proposition \ref{prop:VirasoroNew}. In fact this 2-parameter family of Virasoro fields is "inherited" from $\mathit{F^{\ten 1}}$, with a slight twist, due of course to the twisted vertex algebra isomorphism.
 Allowing multi-locality allows us to trivially construct, in general, new Virasoro representation from existing ones, for instance by the "half the modes, double the charge" trick, see Proposition \ref{prop:doublingcharge}.
 Thus what we see on $\mathit{F^{\ten \frac{1}{2}}}$ is that the 2-point local twisted vertex algebra describing the  correspondence of type D-A  has two distinct types of Virasoro structures (summarized in Theorem \ref{thm:summaryVir}). These structures are distinct in two ways: first, they have different central charges (the main examples have charges correspondingly $\frac{1}{2}$ and 1). The Virasoro fields with central charge $\frac{1}{2}$ are 1-point self-local.  However the Virasoro fields with central charge $1$ are $2$-point local, although they  could be reduced to the usual 1-point locality by a change of variables $z^2$ to $z$ (i.e., they produce genuine Virasoro representations, and not representations of a 2-point Virasoro algebra as in  e.g. \cite{Krich-Nov1}, \cite{Krich-Nov2}, \cite{Schl}, \cite{Cox-Vir}).

Finally, we  construct a $\widehat{\mathcal{D}}$ representation on $\mathit{F^{\ten \frac{1}{2}}}$. $\widehat{\mathcal{D}}$ denotes the universal central extension of the Lie algebra of differential operators on the circle,  see e.g., \cite{Kac}, \cite{WangDual},  although it is very often denoted also by  $W_{1+\infty}$,
We show that  under the $W_{1+\infty}$ action   $\mathit{F^{\ten \frac{1}{2}}}$  is again, not surprisingly,  isomorphic to  $\mathit{F^{\ten 1}}$.

Even though $\mathit{F^{\ten \frac{1}{2}}}$ has a super-vertex algebra structure (i.e., $N=1$ twisted vertex algebra structure), this is not enough to produce the new  Heisenberg  $\mathcal{H}_{\mathbb{Z}}$ representation. To construct the Heisenberg field  we needed to introduce 2-point locality at the least (i.e., the fields we consider on $\mathit{F^{\ten \frac{1}{2}}}$ are allowed to be multi-local,  at both $z=w$ and $z=-w$, see Definition \ref{defn:parity}).  This in turn allows for more descendant fields than in a super vertex algebra. More precisely,  there is a twisted vertex algebra structure on $\mathit{F^{\ten \frac{1}{2}}}$ of order $N$ (see \cite{AngTVA}, \cite{ACJ} for a precise definition of a twisted vertex algebra).
 We will prove in this paper that $\mathit{F^{\ten 1}}\cong\mathit{F^{\ten \frac{1}{2}}}$   as graded vector spaces,  as  $\mathcal{H}_{\mathbb{Z}}$  modules, and as $W_{1+\infty}$ modules, although obviously there should be some difference in terms of the physics structures on these two  spaces. As super vertex algebras, i.e., twisted vertex algebras of order $N=1$,  $\mathit{F^{\ten 1}}$ and  $\mathit{F^{\ten \frac{1}{2}}}$ are not equivalent. But if we introduce an $N=2$ twisted vertex algebra structures both on $\mathit{F^{\ten \frac{1}{2}}}$ and on $\mathit{F^{\ten 1}}$, there is an isomorphism of  twisted vertex algebra structures which constitutes the boson-fermion correspondence of type D-A. If we consider a twisted vertex algebra structure with $N>2$, then once more the twisted vertex algebra structures on  $\mathit{F^{\ten 1}}$ and  $\mathit{F^{\ten \frac{1}{2}}}$ will become inequivalent, as we will show in a follow-up paper.
 This shows that the type of chiral algebra structure on $\mathit{F^{\ten 1}}$ vs $\mathit{F^{\ten \frac{1}{2}}}$ is of great importance, in particular the set of points of locality is a necessary part of the data describing any boson-fermion correspondence. Thus the chiral algebra  structure on $\mathit{F^{\ten \frac{1}{2}}}$, vs the chiral structure on $\mathit{F^{\ten 1}}$ is really what distinguishes these spaces and their physical structure and what defines a boson-fermion correspondence.

\section{Notation and background}
\label{section:background}

We work over the field of complex numbers $\mathbb{C}$. Throughout we assume $N$ is a positive integer.

Some of the following definitions (e.g., that of a field in chiral quantum field theory and normal ordered products) are well known, they can be found for instance in \cite{FLM}, \cite{FHL},  \cite{Kac}, \cite{LiLep} and others, and we include them for completeness:
\begin{defn}
\begin{bf} (Field)\end{bf}\label{defn:field-fin}
 A field $a(z)$ acting on a vector space $V$ is a series of the form
\begin{equation}
a(z)=\sum_{n\in \mathbf{Z}}a_{(n)}z^{-n-1}, \ \ \ a_{(n)}\in
\End(V), \ \ \text{such that }\ a_{(n)}v=0 \ \ \text{for any}\ v\in V, \ n\gg 0.
\end{equation}
\end{defn}
From now throughout the paper let $\epsilon$ be a primitive $N$th root of unity, we will mostly use $N=2$. We also assume that the vector space $V$ (that the fields act on) is a super-vector space, i.e., a $\mathbb{Z}_2$ graded vector space. Although one can define a parity of a  field $a(z)$ through the $\mathbb{Z}_2$ grading of the linear map $a(z): V\to V((z))$, the notion of a parity of a field is mostly only useful  in combination with some notion of locality. Hence we use the following definition:
\begin{defn}(\cite{ACJ})  \label{defn:parity} {\bf ($N$-point self-local fields and parity) }
We say that a field $a(z)$ on a vector space $V$ is {\bf even} and $N$-point self-local at $1, \epsilon, \epsilon^2, \dots, \epsilon^{N-1}$,  if there exist $n_0, n_1, \dots  , n_{N-1}$ such that
\begin{equation}
(z- w)^{n_{0}}(z-\epsilon w)^{n_{1}}\cdots (z-\epsilon^{N-1} w)^{n_{N-1}}[a(z),a(w)] =0.
\end{equation}
In this case we set the {\bf parity} $p(a(z))$ of $a(z)$ to be $0$.
\\
We set $\{a, b\}  =ab +ba$.We say that a field $a(z)$ on $V$ is $N$-point self-local at $1, \epsilon, \epsilon^2, \dots, \epsilon^{N-1}$
and {\bf odd} if there exist $n_0, n_1, \dots , n_{N-1}$ such that
\begin{equation}
(z- w)^{n_{0}}(z-\epsilon w)^{n_{1}}\cdots (z-\epsilon^{N-1} w)^{n_{N-1}}\{a(z),a(w)\}=0.
\end{equation}
In this case we set the {\bf parity} $p(a(z))$ to be $1$. For brevity we will just write $p(a)$ instead of $p(a(z))$. If $a(z)$ is even or odd field, we say that $a(z)$ is homogeneous.\\
Finally,  if $a(z), b(z)$ are homogeneous fields on $V$, we say that $a(z)$ and $b(z)$ are {\it $N$-point mutually local} at $1, \epsilon, \epsilon^2, \dots, \epsilon^{N-1}$
if there exist $n_0, n_1, \dots , n_{N-1} \in \mathbb{Z}_{\geq 0}$ such that
\begin{equation}
(z- w)^{n_{0}}(z-\epsilon w)^{n_{1}}\cdots (z-\epsilon^{N-1} w)^{n_{N-1}}\left(a(z)b(w)-(-1)^{p(a)p(b)}b(w)a(z)\right)=0.
\end{equation}
\end{defn}
For a rational function $f(z,w)$,  with poles only at $z=0$,  $z=\epsilon^i w, \ 0\leq i\leq N-1$, we denote by $i_{z,w}f(z,w)$
the expansion of $f(z,w)$ in the region $\abs{z}\gg \abs{w}$ (the region in the complex $z$ plane outside of all  the points $z=\epsilon^i w, \ 0\leq i\leq N-1$), and correspondingly for
$i_{w,z}f(z,w)$.
Let
\begin{equation}
a(z)_-:=\sum_{n\geq 0}a_nz^{-n-1},\quad a(z)_+:=\sum_{n<0}a_nz^{-n-1}.
\end{equation}
\begin{defn} \label{defn:normalorder} {\bf (Normal ordered product)}
Let $a(z), b(z)$ be homogeneous fields on a vector space $V$. Define
\begin{equation}
:a(z)b(w):=a(z)_+b(w)+(-1)^{p(a)p(b)}b(w)a_-(z).
\end{equation}
One calls this the normal ordered product of $a(z)$ and $b(w)$. We extend by linearity the notion of normal ordered product  to any two fields which are linear combinations of homogeneous fields.
\end{defn}
\begin{remark}
Let  $a(z), b(z)$ be fields on a vector space $V$. Then
$:a(z)b(\epsilon ^i z):$ and $:a(\epsilon ^i z)b( z):$ are well defined fields on $V$ for any $i=0, 1, \dots,  N-1$.
\end{remark}
The mathematical background of the well-known and often used in physics notion of Operator Product Expansion (OPE) of product of two fields for case of usual locality ($N=1$) has been established for example in \cite{Kac}, \cite{LiLep}.
The following lemma extended the mathematical background  to the case of  $N$-point locality:
\begin{lem} (\cite{ACJ}) {\bf (Operator Product Expansion (OPE))}\label{lem:OPE} \\
Suppose $a(z)$, $b(w)$ are {\it $N$-point mutually local}. Then exists fields $c_{jk}(w)$, $j=0, \dots, N-1; k=0, \dots , n_j-1$, such that we have
 \begin{equation}
 \label{eqn:OPEpolcor}
 a(z)b(w) =i_{z, w} \sum_{j=0}^{N-1}\sum_{k=0}^{n_j-1}\frac{c_{jk}(w)}{(z-\epsilon^j w)^{k+1}} + :a(z)b(w):.
 \end{equation}
We call the fields $c_{jk}(w)$, $j=0, \dots, N-1; k=0, \dots , n_j-1$ OPE coefficients.We will write the above OPE as
 \begin{equation}
 a(z)b(w) \sim  \sum_{j=1}^N\sum_{k=0}^{n_j-1}\frac{c_{jk}(w)}{(z-\epsilon_j w)^{k+1}}.
 \end{equation}
 The $\sim $ signifies that we have only written the singular part, and also we have omitted writing explicitly the expansion $i_{z, w}$, which we do acknowledge  tacitly.
 \end{lem}
  \begin{remark}
 Since  the notion of normal ordered product is extended by linearity to any two fields which are linear combinations of homogeneous fields, the Operator Product Expansions formula above applies also  to any two fields which are linear combinations of homogeneous  $N$-point mutually local fields.
\end{remark}
 The OPE expansion of the product of two fields is very convenient, as it completely determines in a very compact manner the commutation relations between the modes of the two fields, and we will use it extensively in what follows.  The OPE expansion in the multi-local case allowed us to  extend the Wick's Theorem (see e.g., \cite{MR85g:81096}, \cite{MR99m:81001}) to the case of multi-locality (see \cite{ACJ}). We further  have the following expansion formula extended to the multi-local case, which we will use extensively in what follows:
 \begin{lem}(\cite{ACJ}) \label{lem:normalprodexpansion}{(\bf Taylor expansion formula for normal ordered products) } \\
Let  $a(z), b(z)$ be  fields on a vector space $V$. Then
\begin{equation}
i_{z, z_0}:a(\epsilon ^i z +z_0)b( z): =\sum _{k\geq 0}\Big(:(\partial_{\epsilon ^i} ^{(k)}a(\epsilon ^i z))b(z):\Big) z_0^k; \quad \text{for any}\ i=0, 1, \dots,  N-1.
\end{equation}
 \end{lem}
\begin{defn}\label{defn:fielddesc} \begin{bf}(The Field Descendants Space  $\mathbf{\mathfrak{FD} \{a^0 (z), a^1 (z),  \dots , a^p(z); N\} }$)\end{bf} \\
Let $a^0 (z), a^1 (z), \dots , a^p(z)$ be given homogeneous fields on a vector space $W$, which are self-local and pairwise $N$-point local with points of locality $1, \epsilon, \dots, \epsilon^{N-1}$. Denote by $\mathfrak{FD} \{a^0 (z), a^1(z), \dots , a_p(z); N\}$ the subspace of all fields on $W$ obtained from the fields $a^0 (z), a^1(z), \dots , a^p(z)$ as follows:
\begin{enumerate}
\item $Id_W, a^0 (z), a^1(z), \dots , a^p(z)\in \mathfrak{FD} \{ a^0 (z), a^1 (z), \dots , a^p(z); N \}$;
\item  If $d(z)\in \mathfrak{FD} \{ a^0 (z), a^1 (z), \dots , a^p(z); N \}$, then $\partial_z (d(z))\in \mathfrak{FD} \{ a^0 (z), \dots , a^p(z); N \}$;
\item  If $d(z)\in \mathfrak{FD} \{ a^0 (z), a^1 (z), \dots , a^p(z); N \}$, then $d(\epsilon^i z)$ are also elements of\\  \mbox{ $\mathfrak{FD} \{a^0 (z), a^1(z), \dots , a^p(z); N\}$} for $i=0,\dots, N-1$;
\item
If $d_1(z), d_2 (z)$ are both in \mbox{$\mathfrak{FD} \{a^0 (z), a^1(z), \dots , a^p(z); N\}$}, then $:d_1(z)d_2(z):$ is also an element of  $\mathfrak{FD} \{a^0 (z), a^1(z), \dots , a^p(z); N\}$, as well as all OPE coefficients in the OPE expansion of  $d_1(z)d_2(w)$.
\item all finite linear combinations of fields in $ \mathfrak{FD} \{ a^0 (z), a^1 (z), \dots , a^p(z); N \}$ are still in \\ $ \mathfrak{FD} \{ a^0 (z), a^1 (z), \dots , a^p(z); N \}$.
\end{enumerate}
\end{defn}

We will not remind here the definition of a twisted vertex algebra as it is rather technical, see instead  \cite{AngTVA}, \cite{ACJ}. A twisted  vertex algebra is a generalization of a super vertex algebra, in the sense that  any super vertex algebra is an $N=1$-twisted vertex algebra, and vice versa. A major difference besides the $N$-point locality is that  in twisted vertex algebras the space of fields $V$ is allowed to be strictly larger than the space of states $W$ on which the fields act (i.e., the field-state correspondence is not necessarily a bijection as for super vertex algebras, but a surjective projection). In that sense a twisted vertex algebra is more similar to a deformed chiral algebra in the sense of \cite{FR}.  We have the following construction theorem for twisted vertex algebras:
\begin{prop}\cite{ACJ}\label{prop:GenF-TVA}
Let $a^0 (z), a^1 (z), \dots a^p(z)$ be given fields on a vector space $W$, which are self-local and pairwise local with points of locality $\epsilon ^i$, $i=1, \dots, N$, where $\epsilon$ is a primitive root of unity. Then any two fields in $\mathfrak{FD} \{a^0 (z), a^1(z), \dots a_p(z); N\}$ are self and mutually $N$-point local. Further, if the fields $a^0 (z), a^1 (z), \dots a^p(z)$ satisfy the conditions for generating fields for a twisted vertex algebra (see \cite{ACJ}), then the space  $\mathfrak{FD} \{a^0 (z), a^1(z), \dots a_p(z); N\}$ has a structure of a twisted vertex algebra.
\end{prop}

\section{Boson-fermion correspondence of type D-A}

  We recall the definitions and notation for the Fock space $\mathit{F^{\ten \frac{1}{2}}}$ and some brief facts about it  as in \cite{Frenkel-BF}, \cite{DJKM3}, \cite{Kac-Lie}, \cite{WangDuality}; we  will follow here in particular the notation of \cite{WangDuality}, \cite{WangDual}. In this section we construct the Heisenberg action and  the various gradings on $\mathit{F^{\ten \frac{1}{2}}}$ and decompose it as a direct sum of irreducible  Heisenberg modules. We also recall some basic facts about the Fock space $\mathit{F^{\ten 1}}$ of a pair of free charged fermions (see e.g.  \cite{Frenkel-BF}, \cite{DJKM3}, \cite{KacRaina}, \cite{Kac})

  We consider a single odd self-local field $\phi ^D(z)$
  \begin{equation}
  \phi ^D(z)=\sum _{n\in \mathbb{Z}+1/2} \phi^D_n z^{-n-1/2} =\sum _{m \in \mathbb{Z}} \phi^D_{-m-\frac{1}{2}} z^{m}
   \end{equation}
   (as in e.g., \cite{KacRaina},  \cite{Triality}, \cite{Kac}, \cite{WangDuality}, \cite{WangDual}, \cite{AngTVA}).
The OPE of $\phi ^D(z)$ is given by
\begin{equation}
\label{equation:OPE-D}
\phi ^D(z)\phi ^D(w)\sim \frac{1}{z-w}.
\end{equation}
This OPE completely determines the commutation relations between the modes $\phi^D_n$, $n\in \mathbb{Z} +1/2$:
\begin{equation}
\label{eqn:Com-D}
\{\phi^D_m,\phi^D_n\}=\delta _{m, -n}1.
\end{equation}
and so the modes generate a Clifford algebra $\mathit{Cl_D}$.
The field $\phi ^D(z)$ is usually  called in the mathematical physics literature  "neutral fermion field".

 The Fock space of the  field $\phi ^D(z)$ is the highest weight module of $\mathit{Cl_D}$ with vacuum vector $|0\rangle $, so that $\phi^D_n|0\rangle=0 \ \text{for} \  n >0$.  It is denoted  by $\mathit{F^{\ten \frac{1}{2}}}$ (see e.g.  \cite{DJKM6}, \cite{Triality}, \cite{Wang},  \cite{WangDual}, \cite{WangDuality}, \cite{WangKac}). This well known Fock space is often called the Fock space of the free neutral fermion (see e.g.   \cite{DJKM6}, \cite{Triality}, \cite{Wang}, \cite{WangKac},  \cite{WangDual}).

\begin{remark}
In \cite{AngTVA} this  neutral fermion Fock space is denoted by $\mathit{F_D}$ as it canonically decomposes into basic modules for the Lie algebra $d_{\infty}$, see \cite{DJKM6}), \cite{Wang}, \cite{WangKac}. Here we use the more popular notation $\mathit{F^{\ten \frac{1}{2}}}$ as used for example in \cite{Wang}, \cite{WangKac}, \cite{WangDual}.
\end{remark}

The space $\mathit{F^{\ten \frac{1}{2}}}$ can be given a super-vertex algebra structure (i.e., with a single point of locality at $z=w$ in the OPEs), as is well known from e.g. \cite{Triality}, \cite{Wang}, \cite{Kac}.
But this super vertex algebra has no Heisenberg element, and thus no boson-fermion correspondence could be constructed on the level of super-vertex algebras, one needs to introduce multi-locality in order to bosonize this neutral fermion Fock space (see Remark \ref{remark:F^1} later on).
As  in any boson-fermion correspondence (a.k.a bosonization), to establish the boson-fermion correspondence  of type D-A  one first constructs the Heisenberg field, and then an invertible map (as fields on   $\mathit{F^{\ten \frac{1}{2}}}$) between $\phi^D(z)$ and certain exponentiated boson fields (which by their definition require a Heisenberg field, see Remark \ref{remark:F^1}. This was done in \cite{AngTVA}, but utilizing the bicharacter construction, which although more general, requires a very different and less known set of tools. Here we will re-prove the necessary information using more standard representation theory tools, in particular we will start by proving  the existence of the Heisenberg field on $\mathit{F^{\ten \frac{1}{2}}}$, and the consequent decomposition of the space $\mathit{F^{\ten \frac{1}{2}}}$ as a direct sum of Heisenberg modules.  We want to emphasize that in contrast to the well-known charged free boson-fermion correspondence which bosonizes  the \textbf{super} vertex algebras $\mathit{F^{\ten 1}}$ ($\mathit{F^{\ten 1}}$ is the Fock space of 1 pair of charged free fermions, see Remark \ref{remark:F^1} later on), and   only requires 1-point locality  at $z=w$, in the case of $\mathit{F^{\ten \frac{1}{2}}}$ the boson-fermion correspondence of type D-A is an isomorphism of \textbf{twisted} vertex algebras, requiring multi-locality, and specifically a field descendent  space  $\mathbf{\mathfrak{FD}}\{ \phi ^D(z); N \} $ with $N\geq 2$, i.e., locality at  $z=\epsilon ^iw$, $i=0, 1, \dots, N-1$  (see \cite{AngTVA}, \cite{ACJ}).
   We will only consider the case of $N=2$ in this paper, i.e., multi-locality only at  $z=w$ and $z=-w$.

 Thus, consider the field $\phi^D(z)$ and its  $N=2$ field descendent  space  $\mathbf{\mathfrak{FD}}\{ \phi ^D(z); 2 \} $ (in particular the descendent field $\phi^D(-z)$), i.e., the fields descendants space generated by allowing $N=2$ multi-locality at  $z=w$ and $z=-w$.
 The space  $\mathbf{\mathfrak{FD}}\{ \phi ^D(z); 2 \} $ contains the normal ordered product $:\phi ^D(z)\phi ^D(-z):$ (for which the point of locality $ -1$, i.e. $z=-w$,  is not removable).
The space  $\mathbf{\mathfrak{FD}}\{ \phi ^D(z); 2 \} $ has a structure of twisted vertex algebra by Proposition \ref{prop:GenF-TVA}.
\begin{prop}\label{prop:Heis} (\cite{AngTVA})
The field $h^D(z)\in \mathbf{\mathfrak{FD}}\{ \phi ^D(z); 2 \} $ given by:
\begin{equation}
\label{eqn:normal-order-h-D}
h^D (z)= \frac{1}{2}:\phi ^D(z)\phi ^D(-z) :
\end{equation}
has only  odd-indexed modes ($h^D (z)=-h^D (-z)$),   $h^D (z)=\sum _{n\in \mathbb{Z}} h_{n} z^{-2n-1}$,  and has OPE with itself given by:
\begin{equation}
\label{eqn:HeisOPEsD}
h^D (z)h^D (w)\sim \frac{zw}{(z^2-w^2)^2} \sim \frac{1}{4}\frac{1}{(z-w)^2} - \frac{1}{4}\frac{1}{(z+w)^2}.
\end{equation}
 Hence its  modes, $h_n, \ n\in \mathbb{Z}$, generate the Heisenberg algebra $\mathcal{H}_{\mathbb{Z}}$ with relations $[h_m,h_n]=m\delta _{m+n,0}1$, \ $m,n$   integers; and we call the field  $h^D(z)$  a Heisenberg field.
 We will denote this  new representation  of $\mathcal{H}_{\mathbb{Z}}$ on $\mathit{F^{\ten \frac{1}{2}}}$ by $r_{D, H}$.
 \end{prop}
In \cite{AngTVA} we gave a proof of this statement using the bicharacter construction. Here we will prove it by Wick's theorem, see e.g. \cite{MR85g:81096}, \cite{MR99m:81001} (the technique  most often used in vertex algebras).
\begin{proof}
From Wick's theorem, writing only the singular OPE terms, we have
\begin{align*}
:\phi ^D(z) \phi ^D(-z): :\phi ^D(w)\phi ^D(-w):
&\sim -\frac{1}{z-w}:\phi ^D(-z)\phi ^D(-w): +\frac{1}{z+w}:\phi ^D(-z)\phi ^D(w): \\
&\quad+\frac{-1}{z+w}:\phi ^D(z)\phi ^D(-w):-\frac{-1}{z-w}:\phi ^D(z)\phi ^D(w): \\
&\quad -\frac{1}{z-w}\cdot \frac{-1}{z-w}+\frac{1}{z+w}\cdot \frac{-1}{z+w}
\end{align*}
Now we apply Taylor expansion formula Lemma \ref{lem:normalprodexpansion}:
\begin{align*}
:\phi ^D(z) \phi ^D(-z): :\phi ^D(w)\phi ^D(-w):
&\sim -\frac{1}{z-w}:\phi ^D(-w)\phi ^D(-w): +\frac{1}{z+w}:\phi ^D(w)\phi ^D(w): \\
&\quad+\frac{-1}{z+w}:\phi ^D(-w)\phi ^D(-w):-\frac{-1}{z-w}:\phi ^D(w)\phi ^D(w): \\
&\quad -\frac{1}{z-w}\cdot \frac{-1}{z-w}+\frac{1}{z+w}\cdot \frac{-1}{z+w}.
\end{align*}
The other terms from the Taylor expansion will produce regular (nonsingular)  terms and thus do not contribute to  the OPE.
It is easily checked that $:\phi ^D(w)\phi ^D(w):=0$ (see e.g. \cite{ACJ}), and we have the desired result. That completes the proof, as the commutation relations between the modes follow directly from the OPE.
\end{proof}
It is also easily seen from continuing the calculation above (attention should be paid that one needs to use here in the multi-local case the Mittag-Leffler expansion (i.e., isolate all the singularities first), not the Laurent expansion in powers of $(z-w)$ only),  that \footnote{This field is a good example for  the nonassociativity of the normal ordered product, as
 $:h^D (z)h^D (z):=\\ :(:\phi ^D(z)\phi ^D(-z):)(:\phi ^D(z)\phi ^D(-z):):$, we have $:\phi ^D(z)\phi ^D(-z):=-:\phi ^D(-z)\phi ^D(z):$, and $:\phi ^D(z)\phi ^D(z):=\\:\phi ^D(-z)\phi ^D(-z):=0$; but nevertheless $:h^D (z)h^D (z):\neq 0$. It is also a good example of the case of shifts (the powers of $z$ adjustment) appearing in a twisted vertex algebra, see \cite{AngTVA}, \cite{ACJ}.}
\begin{equation}
\label{eqn:h-normorder}
:h^D(w)^2:= :h^D(w)h^D(w):=\frac{1}{4}:(\partial_{-w}\phi ^D(-w))\phi ^D(-w): + \frac{1}{4}:(\partial_{w}\phi ^D(w))\phi ^D(w): -\frac{1}{2w}h^D(w)
\end{equation}

Here we have to take a detour and   recall some  information about the best known and and often used charged free boson-fermion correspondence (of type A), its Fock space $\mathit{F^{\ten 1}}$, and some of the well established connections between the neutral fermion Fock space   $\mathit{F^{\ten \frac{1}{2}}}$ and the charged free fermion Fock space $\mathit{F^{\ten 1}}$. $\mathit{F^{\ten 1}}$ is the Fock space of 1 pair of two charged fermions (see e.g. \cite{Frenkel-BF}, \cite{DJKM3}, \cite{Kac}). The two charged fermions are  the fields $\psi^+ (z)$ and $\psi^- (z)$ with only nontrivial operator product expansion (OPE):
\begin{align*}
\psi^+ (z)\psi^- (w)\sim \frac{1}{z-w}\sim \psi ^-(z)\psi ^+(w),
\end{align*}
where the $1$ above denotes the identity map $Id_{\mathit{F^{\ten 1}}}$.
The modes $\psi^+_n$ and $\psi^-_n$, $n\in \mathbb{Z}$ of the fields  $\psi^+ (z)$ and $\psi^- (z)$, which we index as follows (\cite{Kac}):
\begin{equation}
\psi ^+(z) =\sum _{n\in \mathbf{Z}} \psi^+_{n} z^{-n-1}, \quad \psi^- (z) =\sum _{n\in \mathbf{Z}} \psi^-_n z^{-n-1},
\end{equation}
form a Clifford algebra $\mathit{Cl_A}$  with relations
\begin{equation}
\{\psi^+_m,\psi^-_n\}=\delta _{m+n, -1}1, \quad \{\psi^+_m,\psi^+_n\}=\{\psi^-_m,\psi^-_n\}=0.
\end{equation}
The Fock space $\mathit{F^{\ten 1}}$ is the highest weight   representation of $\mathit{Cl_A}$ generated by the vacuum  vector $|0\rangle $,  so that $\psi^+_n|0\rangle=\psi^-_n|0\rangle=0 \ \text{for} \  n\geq 0$. \\
It is well known (see e.g., \cite{Kac}, \cite{FZvi}, \cite{LiLep}) that $\mathit{F^{\ten 1}}$ has a structure of a super vertex algebra (i.e., with a single point of locality at $z=w$ in the OPEs); this vertex algebra  is often called  "charged free fermion vertex algebra". It is also well known (introduced by Igor Frenkel) and extensively used (e.g., \cite{Frenkel-BF}, \cite{DJKM3}, \cite{Kac-Lie}, \cite{WangDuality}, \cite{WangDual} among others) that  $\mathit{F^{\ten 1}}$ is canonically decomposed into basic  modules for the $a_{\infty}$ algebra (hence in \cite{AngTVA} we used the notation $\mathit{F^A}$ instead of $\mathit{F^{\ten 1}}$), and as a direct sum of irreducible Heisenberg $\mathcal{H}_{\mathbb{Z}}$ modules. In the case of $\mathit{F^{\ten 1}}$ the boson-fermion correspondence (of type A) is an isomorphism of super vertex algebras, requiring locality only at $z=w$ (see e.g., \cite{Kac}).  In particular, the well known charged free boson-fermion correspondence is an isomorphism between the super vertex algebra  $\mathit{F^{\ten 1}}$ and the super vertex algebra of the rank one odd lattice vertex algebra $\mathit{B_A}$, with generating fields $e^{\alpha}_A(z)$ and $e^{-\alpha}_A(z)$ (see e.g. \cite{Frenkel-BF}, \cite{DJKM3}, \cite{Kac}):
\begin{align}
\label{eqn:ExponA-1}
e^{\alpha}_A(z)&=\exp (\sum _{n\ge 1}\frac{h^A_{-n}}{n} z^n)\exp (-\sum _{n\ge 1}\frac{h^A_{n}}{n} z^{-n})e^{\alpha}z^{\partial_{\alpha}},\\
\label{eqn:ExponA-2}
e^{-\alpha}_A(z)&=\exp (-\sum _{n\ge 1}\frac{h^A_{-n}}{n} z^n)\exp (\sum _{n\ge 1}\frac{h^A_{n}}{n} z^{-n})e^{-\alpha}z^{-\partial_{\alpha}}.
\end{align}
Two important notes: first,  the above expressions depend on the existence of the  Heisenberg field $h^A(z)=\sum h^A_n z^{-n-1}$, whose modes $h^A_n$ generate the Heisenberg algebra $\mathcal{H}_{\mathbb{Z}}$ and are used in the exponentiated boson fields  $e^{\alpha}(z)$ and $e^{-\alpha}(z)$. In particular, in the case of the charged free boson-fermion correspondence (of type A), the Heisenberg field $h^A(z)$ is constructed from the fermionic fields as (see \cite{Frenkel-BF}, \cite{DJKM3}, \cite{KacRaina}, \cite{Kac})
\begin{equation}
\label{HeisAdepen1}
h^A(z):=:\psi^+ (z)\psi^- (z):.
\end{equation}
Second, the boson-fermion correspondence is the isomorphism $\mathit{F^{\ten 1}}\cong \mathit{B_A}$ of super vertex algebras generated by the invertible map
\begin{equation}
\psi^+(z)\mapsto e^{\alpha}_A(z), \quad \psi^-(z)\mapsto e^{-\alpha}_A(z).
\end{equation}
To establish this isomorphism, then, there are  two important necessary ingredients: first, one needs the  Heisenberg field $h^A(z)$ in terms of the generating fields on $\mathit{F^{\ten 1}}$, and second, the  charged decomposition of $\mathit{F^{\ten 1}}$ as a direct sum of irreducible Heisenberg modules. In particular (\cite{DJKM3}, \cite{KacRaina})
\begin{equation}\label{Aequiv}
\mathit{F^{\ten 1}}\cong\oplus _{m\in \mathbb{Z}} \mathit{F_{(m)}^{\ten 1}}\cong \oplus _{m\in \mathbb{Z}} B_m,  \quad \text{where} \quad \mathit{F_{(m)}^{\ten \frac{1}{2}}}\cong B_m,\quad
B_m \cong  \mathbb{C}[x_1, x_2, \dots , x_n, \dots ],\quad \forall \ m\in \mathbb{Z}.
\end{equation}
Here $\mathit{F_{(m)}^{\ten 1}}$ is  the subspace of charge $m$, see \cite{KacRaina} (we changed the notation slightly), and it is a
 well known fact (see e.g. \cite{KacRaina}, \cite{FLM}) that any irreducible highest weight module of the Heisenberg
algebra $\mathcal{H}_{\mathbb{Z}}$ is isomorphic to the polynomial algebra with infinitely many variables $\mathit{B_\lambda}\cong \mathbb{C}[x_1,
x_2, \dots , x_n, \dots ]$ via the action:
\begin{equation}
h_n\mapsto \partial _{x_{n}}, \quad h_{-n} \mapsto
nx_n\cdot, \quad \text{for any} \ \ n\in \mathbb{N}, \quad h_0\mapsto \lambda\cdot.
\end{equation}
In fact we can introduce an arbitrary re-scaling $s_n\neq 0, \ s_n\in \mathbb{C}$, for $n\neq 0$ only, so that
\begin{equation}
h_n\mapsto s_n\partial _{x_{n}}, \quad h_{-n} \mapsto
s_n^{-1}nx_n\cdot, \quad \text{for any} \ \ n\in \mathbb{N}, \quad h_0\mapsto \lambda\cdot.
\end{equation}
Here we will build the boson-fermion correspondence of type D-A, which is an isomorphism of twisted vertex algebras, by constructing these two ingredients. We already established the existence of the Heisenberg field $h^D(z)$ on $\mathbf{\mathit{F^{\ten \frac{1}{2}}}}$, next we will prove the analogue of the charged decomposition for $\mathbf{\mathit{F^{\ten \frac{1}{2}}}}$.

\begin{remark}\label{remark:F^1}
We want to address here a common misconception that has been brought to our attention by certain  comments we received.
It is well known (introduced first by Igor Frenkel in \cite{Frenkel-BF} and extensively used afterwards) that
\begin{equation*}
\mathit{F^{\ten \frac{1}{2}}}\ten \mathit{F^{\ten \frac{1}{2}}} \cong \mathit{F^{\ten 1}};\quad \text{via}
\quad \frac{\phi^{D, 1}(z) +\phi^{D, 2}(z)}{\sqrt{2}}\mapsto \psi^+(z), \quad \frac{\phi^{D, 1}(z) -\phi^{D, 2}(z)}{\sqrt{2}}\mapsto \psi^-(z);
\end{equation*}
here the field $\phi^{D, 1}(z)$ ($\phi^{D, 2}(z)$) is identified as the field acting on the first (correspondingly second) component of $\mathit{F^{\ten \frac{1}{2}}}\ten \mathit{F^{\ten \frac{1}{2}}}\cong \mathit{F^{\ten 1}}$. But that doesn't mean that this identification gives a bosonization of $\mathbf{\mathit{F^{\ten \frac{1}{2}}}}$, for two reasons as the notes above indicate. First, this identification does \textbf{not}  produce a Heisenberg free boson field on $\mathbf{\mathit{F^{\ten \frac{1}{2}}}}$, i.e., in terms of the generating field $\phi^D(z)$, and consequently one can't write exponentiated boson fields, as the Heisenberg field is part of their definition as in \eqref{eqn:ExponA-1} and \eqref{eqn:ExponA-2}.
Second, the fact that one can write the field $\phi^{D, 1}(z)$ as $\phi^{D, 1}(z)=\frac{1}{\sqrt{2}}(e^{\alpha}_A(z)+e^{-\alpha}_A(z))$ on the space $\mathit{F^{\ten 1}}$ doesn't mean that one can write $\phi^D(z)= \frac{1}{\sqrt{2}}(e^{\alpha}_A(z)+e^{-\alpha}_A(z))$  on the space
$\mathbf{\mathit{F^{\ten \frac{1}{2}}}}$, as  the fields $e^{\alpha}(z)$ and $e^{-\alpha}(z)$    are not fields (do not act) on each component $\mathbf{\mathit{F^{\ten \frac{1}{2}}}}$ separately, but on $\mathit{F^{\ten 1}} \cong  \mathit{F^{\ten \frac{1}{2}}}\ten \mathit{F^{\ten \frac{1}{2}}}$!  Thus, both required ingredients for bosonization are unobtainable  if one just  uses the \textbf{super vertex algebra structure} on $\mathbf{\mathit{F^{\ten \frac{1}{2}}}}$. There is no Heisenberg field at all in the \textbf{super vertex algebra} $\mathbf{\mathfrak{FD}}\{ \phi ^D(z); 1 \} $ generated by $\phi^D(z)$.  Consequently, even though there is a well-established super vertex algebra isomorphism  $\mathit{F^{\ten 1}} \cong  \mathit{F^{\ten \frac{1}{2}}}\ten \mathit{F^{\ten \frac{1}{2}}}$, that doesn't bosonise the Fock space  $\mathbf{\mathit{F^{\ten \frac{1}{2}}}}$.
As we saw in Proposition \ref{prop:Heis}, one needs to go to the \textbf{twisted vertex algebra} descendant fields $\mathbf{\mathfrak{FD}}\{ \phi ^D(z); N \} $, $N\geq 2$ to construct the field $h^D(z)$, which  is 2-point local at both $z=w$ and $z=-w$.
  Thus in contrast to the charged free boson-fermion correspondence between the \textbf{super} vertex algebras $\mathit{F^{\ten 1}}$ and $\mathit{B_A}$  (which only requires locality only at $z=w$), in the case of $\mathit{F^{\ten \frac{1}{2}}}$, as we  recall from \cite{AngTVA}, the boson-fermion correspondence of type D-A which bosonizes $\mathit{F^{\ten \frac{1}{2}}}$ is an isomorphism of \textbf{twisted} vertex algebras, requiring multi-locality, at least  at both $z=w$ and $z=-w$ (see \cite{AngTVA}, \cite{ACJ}).
\end{remark}

Although some properties of the Heisenberg field $h^D(z)$ (Proposition \ref{prop:Heis}) on $\mathit{F^{\ten \frac{1}{2}}}$ are similar to the properties of the Heisenberg field $h^A(z)$ on $\mathit{F^{\ten 1}}$, certainly in terms of the Lie algebra structure of the modes, there are differences, some of which less obvious than the multi-locality.
In particular, in the case of the charged free boson-fermion correspondence (of type A), for the Heisenberg field $h^A(z)$, \eqref{HeisAdepen1},  one has
 (see e.g. \cite{Kac}):
\begin{equation}
\label{HeisAdepen2}
:h^A(w)^2:= :h^A(w)h^A(w):=:(\partial_w \psi^+ (w))\psi^- (w): +:(\partial_w \psi^- (w))\psi^+ (w):.
\end{equation}
Hence in the case of type A (usual 1 point locality at $z=w$) the field $:h^A(w)^2:$ is a linear combination of the only two independent conformal weight two fields $:(\partial_w \psi^+ (w))\psi^- (w):$ and $:(\partial_w \psi^- (w))\psi^+ (w):$ (the generating fields $ \psi^+ (w)$ and $\psi^- (w)$ are each of weight  $\frac{1}{2}$). In contrast, for the type D (multi-locality at $z=w$ and $z=-w$) we have the 4 standard  weight 2 fields (the generating field $\phi ^D(w)$ is of weight $\frac{1}{2}$):
 \begin{equation}
:\partial_{z} \phi ^D (z) \phi ^D(z):,\  :\partial_{-z} \phi ^D (-z) \phi ^D(-z):, \ :\partial_{z} \phi ^D (z) \phi ^D(-z):,\  :\partial_{-z} \phi ^D (-z) \phi ^D(z):.
\end{equation}
In particular, for $:h^D(w)^2:$ we have \eqref{eqn:h-normorder}.  But for $\partial_zh^D (z)$ we have
\begin{equation}\label{eqn:hderiv}
\partial_zh^D (z)= \frac{1}{2}\left(:\partial_{z} \phi ^D (z) \phi ^D(-z): +:\partial_{-z} \phi ^D (-z) \phi ^D(z):\right).
\end{equation}
There are of course infinitely many weight 2 fields on $\mathit{F^{\ten \frac{1}{2}}}$, as we saw already in \eqref{eqn:h-normorder},  but the 4 standard fields above also have the special property of having 4th order poles in their OPEs, which will be needed for the construction of the Virasoro fields in the next section. One has to use caution when dealing with the multi-local fields, especially in view of the shifts (the powers of $z$ corrections) that could appear, as we already saw in \eqref{eqn:h-normorder}.

We have established that  $\mathit{F^{\ten \frac{1}{2}}}$ has a Heisenberg field $h^D(z)$, and hence an action of  the Heisenberg algebra $\mathcal{H}_{\mathbb{Z}}$.  If we introduce a normal ordered product $:\phi^D_m \phi^D_n:$ on the modes $\phi^D_m$ of the field $\phi^{D}(z)$, compatible with the normal ordered product of fields (Definition \ref{defn:normalorder}), we have to have
\begin{equation}
\label{eqn:normord}
:\phi^{D}(z)\phi^{D}(w): =\sum _{m, n\in \mathbb{Z}+1/2} :\phi^D_m \phi^D_n: z^{-m-1/2}w^{-m-1/2}= \sum _{m,n\in \mathbb{Z}} :\phi^D_{-m-\frac{1}{2}} \phi^D_{-n-\frac{1}{2}}:z^{m}w^{n},
\end{equation}
and thus for $m, n\in \mathbb{Z}$ we have
\begin{align}
\label{eqn:normord1}
:\phi^D_{-m-\frac{1}{2}}& \phi^D_{-n-\frac{1}{2}}: =\phi^D_{-m-\frac{1}{2}}\phi^D_{-n-\frac{1}{2}} \quad \text{for}\  m+n\neq 1\\
\label{eqn:normord1}
:\phi^D_{-m-\frac{1}{2}}& \phi^D_{-n-\frac{1}{2}}: =\phi^D_{-m-\frac{1}{2}}\phi^D_{-n-\frac{1}{2}} -1=-\phi^D_{-n-\frac{1}{2}}\phi^D_{-m-\frac{1}{2}} \quad \text{for}\ m+n= -1, n\geq 0,\\
\label{eqn:normord1}
:\phi^D_{-m-\frac{1}{2}}& \phi^D_{-n-\frac{1}{2}}: =\phi^D_{-m-\frac{1}{2}}\phi^D_{-n-\frac{1}{2}} \quad \text{for}\ m+n= -1, m\geq 0.
\end{align}
Hence we can express the modes of the field  $h^D (z)=\sum _{n\in \mathbb{Z}} h_{n} z^{-2n-1}$ as follows
\begin{equation}
\label{eq:Heismodes}
h_n=\frac{1}{2}\sum _{p\in \mathbb{Z}+1/2}(-1)^{p-1/2}:\phi^D_{p} \phi^D_{2n-p}: =\frac{1}{2}\sum _{i\in \mathbb{Z}}(-1)^{i+1}:\phi^D_{-i-\frac{1}{2}} \phi^D_{2n+1+ i-\frac{1}{2}}:
\end{equation}
In order to decompose $\mathit{F^{\ten \frac{1}{2}}}$ as a direct sum of irreducible highest modules for $\mathcal{H}_{\mathbb{Z}}$ we need to introduce various gradings on  $\mathit{F^{\ten \frac{1}{2}}}$.
Since $\mathit{F^{\ten \frac{1}{2}}}$  is the highest weight module of $\mathit{Cl_D}$, defined as above,
$\mathit{F^{\ten \frac{1}{2}}}$  has basis
\begin{equation}
\{ \phi^D_{-n_k-\frac{1}{2}}\dots \phi^D_{-n_2-\frac{1}{2}}\phi^D_{-n_1-\frac{1}{2}}|0\rangle, |0\rangle \  \arrowvert \ \ n_k>\dots >n_2>n_1\geq 0, \ n_i\in \mathbb{Z}, \ i=1, 2, \dots, k\}
\end{equation}
The space $\mathit{F^{\ten \frac{1}{2}}}$ has  a $\mathbb{Z}_2$ grading given by  $ k\  mod\  2$ (it is often used see e.g. \cite{WangKac},  \cite{WangDuality}, \cite{WangDual}, as the super-grading for the super vertex algebra structure),
\[
\mathit{F^{\ten \frac{1}{2}}}=\mathit{F_{\bar{0}}^{\ten \frac{1}{2}}}\oplus \mathit{F_{\bar{1}}^{\ten \frac{1}{2}}},
\]
where $\mathit{F_{\bar{0}}^{\ten \frac{1}{2}}}$ (resp. $\mathit{F_{\bar{1}}^{\ten \frac{1}{2}}}$) denote the even (resp. odd) components of $\mathit{F^{\ten \frac{1}{2}}}$. This $\mathbb{Z}_2$ grading can be extended  to a $\mathbb{Z}_{\geq 0}$ grading $\tilde{L}$, called "length", by setting
\begin{equation}
\tilde{L} (\phi^D_{-n_k-\frac{1}{2}}\dots \phi^D_{-n_2-\frac{1}{2}}\phi^D_{-n_1-\frac{1}{2}}|0\rangle)=k.
\end{equation}

We introduce a second, $\mathbb{Z}$, grading $dg$ on $\mathit{F^{\ten \frac{1}{2}}}$ by assigning $dg(|0\rangle)=0$ and
\begin{equation*}
dg(\phi^D_{-n_k-\frac{1}{2}}\dots \phi^D_{-n_2-\frac{1}{2}}\phi^D_{-n_1-\frac{1}{2}}|0\rangle)=\#\{i=1, 2, \dots, k|n_i=\text{odd}\}\ -\#\{i=1, 2, \dots, k|n_i=\text{even}\}.
\end{equation*}
Denote the homogeneous component of degree $dg=n \in \mathbb Z$ by $\mathit{F_{(n)}^{\ten \frac{1}{2}}}$, hence as vector spaces we have
\begin{equation}
\mathit{F^{\ten \frac{1}{2}}} = \oplus _{n\in \mathbb{Z}} \mathit{F_{(n)}^{\ten \frac{1}{2}}}.
\end{equation}
Introduce  also the special vectors $v_n\in  \mathit{F_{(n)}^{\ten \frac{1}{2}}}$ defined by
\begin{align}
&v_0=|0\rangle \in  \mathit{F_{(0)}^{\ten \frac{1}{2}}};\\
&v_n=\phi^D_{-2n+1-\frac{1}{2}}\dots \phi^D_{-3-\frac{1}{2}}\phi^D_{-1-\frac{1}{2}}|0\rangle \in  \mathit{F_{(n)}^{\ten \frac{1}{2}}}, \quad \text{for}\ n>0;\\
&v_{-n}=\phi^D_{-2n+2-\frac{1}{2}}\dots \phi^D_{-2-\frac{1}{2}}\phi^D_{-\frac{1}{2}}|0\rangle\in  \mathit{F_{(-n)}^{\ten \frac{1}{2}}}, \quad \text{for}\ n>0.
\end{align}
Note that the vectors $v_n\in \mathit{F_{(n)}^{\ten \frac{1}{2}}}$ have minimal length $\tilde{L}=|n|$  among the vectors within $\mathit{F_{(n)}^{\ten \frac{1}{2}}}$, and they are in fact the unique (up-to a scalar) vectors minimizing  the length $\tilde{L}$, such that the index $n_k$ is minimal too ($n_k$ is identified from the smallest  index appearing). One can view each of the vectors $v_n$ as a vacuum-like vector in $\mathit{F_{(n)}^{\ten \frac{1}{2}}}$, and the $dg$ grading as the analogue of the "charge" grading in $\mathit{F_{(n)}^{\ten 1}}$.

Note also that the super-grading derived from the second grading $dg$ is compatible with the super-grading derived from the first grading $\tilde{L}$,  $dg \ mod\  2 =\tilde{L}\  mod\  2$.

Lastly, we want to also define a grading $deg_h$ on each of the components $\mathit{F_{(n)}^{\ten \frac{1}{2}}}$ for each $n\in \mathbb{Z}$. Consider a monomial vector $v=\phi^D_{-n_k-\frac{1}{2}}\dots \phi^D_{-n_2-\frac{1}{2}}\phi^D_{-n_1-\frac{1}{2}}|0\rangle$ from $\mathit{F_{(n)}^{\ten \frac{1}{2}}}$. One can view this vector as an "excitation" from the vacuum-like vector $v_n$, and count the $n_i$ that  should have been in $v$ as compared to $v_n$, and the also the  $n_i$ that  should not have been in $v$ as compared to $v_n$. Thus the grading $deg_h$ (one can think of it as "energy") is defined as
\begin{align*}
deg_h (v)& =\sum \{\floor[\big ]{\frac{n_l+1}{2}}\  |\  n_l\  \text{that \ should \ have \ been \ there\ but \ are \ not}\}\\
&\quad +\sum \{\floor[\big ]{\frac{n_l+1}{2}}\ |\ n_l\ \text{that \ should \ not have \ been \ there\ but\ are}\};
\end{align*}
here $\floor[\big ]{x}$ denotes the floor function, i.e.,   $\floor[\big ]{x}$ denotes the largest integer smaller or equal to  $x$.

Example: consider $\phi^D_{-\frac{5}{2}}|0\rangle$. We have $\tilde{L}(\phi^D_{-\frac{5}{2}}|0\rangle)=1$, and
\[
\phi^D_{-\frac{5}{2}}|0\rangle=\phi^D_{-2-\frac{1}{2}}|0\rangle
\]
and so  $dg(\phi^D_{-\frac{5}{2}}|0\rangle)=-1$, $\phi^D_{-\frac{5}{2}}|0\rangle\in \mathit{F_{(-1)}^{\ten \frac{1}{2}}}$. The vacuum-like vector of $\mathit{F_{(-1)}^{\ten \frac{1}{2}}}$ is $v_{-1}=\phi^D_{-\frac{1}{2}}|0\rangle$. Hence
\[
deg_h (\phi^D_{-\frac{5}{2}}|0\rangle) =\floor[\big ]{\frac{2+1}{2}} +\floor[\big ]{\frac{0+1}{2}}=1.
\]
Another example: consider $v=\phi^D_{-\frac{11}{2}}\phi^D_{-\frac{7}{2}}\phi^D_{-\frac{1}{2}}|0\rangle$. We have $\tilde{L}(v)=3$, and
\[
v=\phi^D_{-5-\frac{1}{2}}\phi^D_{-3-\frac{1}{2}}\phi^D_{-0-\frac{1}{2}}|0\rangle
\]
and so  $dg(v)=1$, $v\in \mathit{F_{(1)}^{\ten \frac{1}{2}}}$.  The vacuum-like vector in $\mathit{F_{(1)}^{\ten \frac{1}{2}}}$ is $v_{1}=\phi^D_{-1-\frac{1}{2}}|0\rangle$, and so finally,
\[
deg_h (v) =\floor[\big ]{\frac{5+1}{2}} +\floor[\big ]{\frac{3+1}{2}}+ \floor[\big ]{\frac{0+1}{2}}+ \floor[\big ]{\frac{1+1}{2}}=4.
\]
We want to note the following fact, which is similar to Proposition 4.1 of \cite{KacRaina}:
\begin{lem}\label{lem:countdeg}
Denote by  $\mathit{F_{(n, k)}^{\ten \frac{1}{2}}}$ the linear span of all vectors of grade $deg_h =k$ in $\mathit{F_{(n)}^{\ten \frac{1}{2}}}$. We have
\begin{itemize}
\item $\mathit{F_{(n, 0)}^{\ten \frac{1}{2}}}=\mathbb{C}v_n$;
\item $\mathit{F_{(n)}^{\ten \frac{1}{2}}}=\oplus_{k\in \mathbb{Z}_{+}}\mathit{F_{(n, k)}^{\ten \frac{1}{2}}}$;
\item $dim \mathit{F_{(n, k)}^{\ten \frac{1}{2}}}=p(k)$, where $p(k)$ stands for the number of partitions of $k\geq 0$ into a sum of non-increasing positive integers, with $p(0)=1$.
\end{itemize}
\end{lem}
\begin{proof}
The first two points are obvious. For the third part, we will only prove that $dim \mathit{F_{(0, k)}^{\ten \frac{1}{2}}}=p(k)$, i.e., we will consider the $deg_h=k$ component in $\mathit{F_{(0)}^{\ten \frac{1}{2}}}$, the proof is similar for any $\mathit{F_{(n)}^{\ten \frac{1}{2}}}$, $n\in \mathbb{Z}$.
Consider a partition $\lambda_k=(k_0, k_1, \dots, k_{l-1})$ of $k$ into $l$ parts, i.e.
\[
k=k_0+\dots k_{l-1}, \quad k_0\geq k_1\dots \geq k_{l-1}\geq 1.
\]
The vacuum in $\mathit{F_{(0)}^{\ten \frac{1}{2}}}$ is $|0\rangle$. Assign a unique element $v_{\lambda_{k}}$ to this partition $\lambda_k$ by
\[
v_{\lambda_{k}}=\phi^D_{-2(k_0-0)+\frac{1}{2}}\phi^D_{-2(k_1-1)+\frac{1}{2}}\dots \phi^D_{-2(k_{l-1}-l+1)+\frac{1}{2}} \phi^D_{-2(l-1)-\frac{1}{2}}\phi^D_{-2(l-2)-\frac{1}{2}}\dots \phi^D_{-\frac{1}{2}}  |0\rangle
\]
It is clear that  $v_{\lambda_{k}}\in \mathit{F_{(0, k)}^{\ten \frac{1}{2}}}$.
\end{proof}
Next we  establish the  decomposition of $\mathit{F^{\ten \frac{1}{2}}}$ as a direct sum of irreducible highest modules for $\mathcal{H}_{\mathbb{Z}}$:
\begin{prop} \label{prop:heisdecomp}\cite{AngTVA}
The neutral fermion Fock space  $\mathit{F^{\ten \frac{1}{2}}}$ decomposes into irreducible highest weight modules for  $\mathcal{H}_{\mathbb{Z}}$ under the action $r_{D, H}$ as follows:
\begin{equation}
\mathit{F^{\ten \frac{1}{2}}} = \oplus _{m\in \mathbb{Z}} \mathit{F_{(m)}^{\ten \frac{1}{2}}}\cong \oplus _{m\in \mathbb{Z}} B_m,  \quad \text{where} \quad \mathit{F_{(m)}^{\ten \frac{1}{2}}}\cong B_m,\quad
B_m \cong  \mathbb{C}[x_1, x_2, \dots , x_n, \dots ],\quad \forall \ m\in \mathbb{Z}.
\end{equation}
\end{prop}
In \cite{AngTVA} we gave a brief proof of this statement using the bicharacter construction, here we will give an alternative proof using methods more standard to representation  theory.
\begin{proof}
The proof will consist of three steps: First, we show that each  component $\mathit{F_{(n)}^{\ten \frac{1}{2}}}$ ($n\in \mathbb{Z}$ is a submodule for the  Heisenberg algebra $\mathcal{H}_{\mathbb{Z}}$. Second, we show that each component $\mathit{F_{(n)}^{\ten \frac{1}{2}}}$ ($n\in \mathbb{Z}$ is a highest weight module for the Heisenberg algebra $\mathcal{H}_{\mathbb{Z}}$, with highest vector  $v_n$ and highest weight $n$. And third, we show that the monomials
$h_{-k_{l}}\dots h_{-k_{1}}v_n$ for $0\leq k_1\leq \dots \leq k_l$ span the component $\mathit{F_{(n)}^{\ten \frac{1}{2}}}$. The second and third step together imply (see e.g. Proposition 2.1 of \cite{KacRaina}, or \cite{FLM})  that $\mathit{F_{(n)}^{\ten \frac{1}{2}}}\cong B_n \cong  \mathbb{C}[x_1, x_2, \dots , x_n, \dots ]$

For the first step, it is clear from \eqref{eq:Heismodes} that the action  $r_{D, H}$ preserves the grading $dg$ on $\mathit{F^{\ten \frac{1}{2}}}$, as each of the summands in any of the $h_n$ preserves the $dg$ grading. Thus, each homogeneous $dg$ component is a submodule for the  Heisenberg algebra $\mathcal{H}_{\mathbb{Z}}$.

For the second step, note that in \cite{ACJ2} we constructed a representation of the infinite rank Lie algebra $a_{\infty}$ on $\mathit{F^{\ten \frac{1}{2}}}$, a representation that actually restricts to  the representation $r_{D, H}$ of the Heisenberg algebra $\mathcal{H}_{\mathbb{Z}}$ (note that of course $\mathcal{H}_{\mathbb{Z}}$ is a subalgebra of $a_{\infty}$). Since in \cite{ACJ2} we proved that this  representation of $a_{\infty}$ is a highest weight representation on each component $\mathit{F_{(n)}^{\ten \frac{1}{2}}}$ (see Proposition 3.5), we will just use it here by noting that of course each $h_n$ for $n>0$ is an element of $a_{\infty}^+$.

For the third step, notice that from \eqref{eq:Heismodes} we have
\[
h_{-m}=\frac{1}{2}\sum _{i\in \mathbb{Z}}(-1)^{i+1}:\phi^D_{-i-\frac{1}{2}} \phi^D_{-2m+1+ i-\frac{1}{2}}:
\]
and thus if an element $v$ is of degree $deg_h(v)=q$, $v\in \mathit{F_{(n, q)}^{\ten \frac{1}{2}}}$, then $deg_h(h_{-m}v)=q+\floor[\big ]{\frac{2m}{2}}=q+m$, i.e., $h_{-m}(\mathit{F_{(n, q)}^{\ten \frac{1}{2}}})\subseteq \mathit{F_{(n, q+m)}^{\ten \frac{1}{2}}}$. Hence consider the monomial elements
\[
h_{-k_{l}}\dots h_{-k_{1}}v_n, \quad 0<k_1\leq \dots \leq k_l;
\]
over the different partitions $0<k_1\leq \dots \leq k_l$ of $k=k_1+\dots +k_l$. These monomial elements are linearly independent, and there is $p(k)$ of them. Moreover, for each partition $0<k_1\leq \dots \leq k_l$ of $k=k_1+\dots +k_l$ we have
\[
h_{-k_{l}}\dots h_{-k_{1}}v_n \in \mathit{F_{(n, k)}^{\ten \frac{1}{2}}}.
\]
Since from Lemma \ref{lem:countdeg} the dimension of $\mathit{F_{(n, k)}^{\ten \frac{1}{2}}}$ is precisely $p(k)$, then these monomial elements span $\mathit{F_{(n, k)}^{\ten \frac{1}{2}}}$. This concludes the proof of the third step, and the proof of the proposition.
\end{proof}
\begin{cor}\label{cor:Heisequivalence}
As Heisenberg  $\mathcal{H}_{\mathbb{Z}}$ modules
\[
 \mathit{F^{\ten \frac{1}{2}}}\cong  \mathit{F^{\ten 1}}
 \]
 \end{cor}
\begin{remark}
In some instances it has been (erroneously) remarked that the  famous charged free boson-fermion correspondence (what we call correspondence of type A) is equivalent/nothing-more then the isomorphism  as  Heisenberg $\mathcal{H}_{\mathbb{Z}}$ modules given by \eqref{Aequiv}.
 But as we see from Corollary \ref{cor:Heisequivalence}, from the standpoint of Heisenberg $\mathcal{H}_{\mathbb{Z}}$ modules  $\mathit{F^{\ten 1}}$ and $\mathit{F^{\ten \frac{1}{2}}}$ are isomorphic, although obviously there should be some difference in terms of the physics structures on those spaces. The difference is that the boson-fermion correspondence of type D-A is  a  twisted vertex algebra structure isomorphism, as opposed to the  boson-fermion correspondence of type A which is a super vertex algebra structure isomorphism.
 (See \cite{AngTVA} for more on the structure of the twisted vertex algebra with space of states $\mathit{F^{\ten \frac{1}{2}}}$ and space of fields  $ \mathbf{\mathfrak{FD}}\{ \phi ^D(z); 2 \}$).
 So what we have shown once more here (see also \cite{ACJ2}) is that the boson-fermion correspondences are more than just maps (or isomorphisms) between certain Lie algebra modules: a boson-fermion correspondence is first and foremost an isomorphism between two different   chiral field theories, one fermionic (expressible in terms of free fermions and their descendants), the other bosonic (expressible in terms  of exponentiated bosons).
\end{remark}

\section{Virasoro fields on the Fock space $\mathit{F^{\ten \frac{1}{2}}}$}

We now turn to constructing Virasoro fields on the Fock space $\mathit{F^{\ten \frac{1}{2}}}$.  Recall the  well-known  Virasoro algebra $Vir$, the central extension of the complex polynomial vector fields on the circle. The Virasoro  algebra $Vir$ is the Lie algebra with generators $L_n$, $n\in \mathbb{Z}$, and central element $C$, with
commutation relations
\begin{equation}
\label{eqn:VirCRs}
[L_m, L_n] =(m-n)L_{m+n} +\delta_{m, -n}\frac{(m^3-m)}{12}C; \quad [C, L_m]=0, \ m, n\in \mathbb{Z}.
\end{equation}
Equivalently, the Virasoro field
$L(z): =\sum _{n\in \mathbb{Z}} L_{n} z^{-n-2}$
has OPE with itself given by:
\begin{equation}
\label{eqn:VirOPEs}
L(z)L(w)\sim \frac{C/2}{(z-w)^4} + \frac{2L(w)}{(z-w)^2}+ \frac{\partial_{w}L(w)}{(z-w)}.
\end{equation}
\begin{defn}\label{defn:VirStr}
We say that a twisted vertex algebra with a space of fields $V$ has a Virasoro field if there is field in $V$ such that its modes are the generators of  the Virasoro algebra $Vir$.
\end{defn}
We want to mention that the Virasoro field is of conformal weight 2 (for a precise definition of conformal weight see e.g. \cite{FLM}, \cite{Kac}, \cite{LiLep}).

We want to recall first the known facts about the Virasoro fields on  $\mathit{F^{\ten 1}}$ and correspondingly $\mathit{F^{\ten \frac{1}{2}}}$.
The Fock space  $\mathit{F^{\ten 1}}$, the fermionic side of the boson-fermion correspondence of type A, see above, which is generated from the two odd 1-point local fields $\psi ^+ (z)$ and $\psi ^- (z)$, has a one-parameter family of Virasoro fields with central charge $-12\lambda^2 +12\lambda -2$ (see e.g. \cite{Kac}, Chapter 5):
\begin{align}
\label{eqn:Vir-typeA}
L^{A, \lambda} (z)&=\frac{1}{2}:h^A (z)^2: +(\frac{1}{2} -\lambda )\partial_z h^A (z)\\
&=(1-\lambda):(\partial_z\psi ^+ (z))\psi ^- (z):+\lambda :(\partial _z \psi ^- (z))\psi ^+ (z):,
\end{align}
where $ h^A (z)=:\psi ^+ (z)\psi ^- (z):$ is the Heisenberg field for the correspondence of type A, $\lambda \in \mathbb{C}$.

$ h^A (z)$ is of conformal weight 1 (roughly speaking the normal order products and the derivatives behave as expected with respect to the conformal weight).  We would like to underline that the two components of the Virasoro field come from the  two standard   conformal-weight-2-fields
$:(\partial_z\psi ^+ (z))\psi ^- (z):$ and $:(\partial _z \psi ^- (z))\psi ^+ (z):$.   Their difference equals  $\partial_z h^A (z)$, and their sum  $:h^A (z)^2:$ (\eqref{HeisAdepen2}). Hence for the  correspondence  of type A the Virasoro field $L^{A, \lambda} (z)$ is  a linear combination of the normal ordered products
$:(\partial_z\psi ^+ (z))\psi ^- (z):$ and $:(\partial _z \psi ^- (z))\psi ^+ (z):$.  In fact, there is  a more general family of Virasoro fields $L^{A, \lambda, b} (z)$ for any $\lambda, b\in \mathbb{C}$, with the same central charge $-12\lambda^2 +12\lambda -2$ (\cite{Iohara}):
\begin{equation}
\label{eqn:Vir-typeA2}
L^{A, \lambda, b} (z)=\frac{1}{2}:h^A (z)^2: +(\frac{1}{2} -\lambda )\partial_z h^A (z) -\frac{b}{z}h^A(z) +\frac{b(b-2\lambda +1)}{2z^2}.
\end{equation}
Note that due to the shift $\frac{b}{z}$ the two-parameter field $L^{A, \lambda, b} (z)$ can  not be a vertex operator in a super vertex algebra, but the shifts have to be allowed in a twisted vertex algebra, as they are inevitable.

Now for the Fock space  $\mathit{F^{\ten \frac{1}{2}}}$: it is well known $\mathit{F^{\ten \frac{1}{2}}}$ carry a super vertex operator algebra structure (conformal super vertex algebra in the language of \cite{Kac} for instance), and as such is also a module for the Virasoro algebra with central charge $c=\frac{1}{2}$ (see for example \cite{Triality}, \cite{Wang}, \cite{WangDual}). The well known Virasoro field $L^{1/2}(z)$ on $\mathit{F^{\ten \frac{1}{2}}}$ is given by
\begin{equation}
\label{eqn:Vir1/2}
L^{1/2}(z)=\frac{1}{2}: \partial_z \phi ^D (z) \phi ^D(z):.
\end{equation}
Hence, as super vertex algebra, $\mathit{F^{\ten \frac{1}{2}}}$ has a conformal vector $w_c=\frac{1}{2}\phi^D_{-3/2}\phi^D_{-1/2}|0\rangle$  (see e.g.  \cite{Triality}, \cite{Wang}, \cite{WangDual}).

Before proceeding to the new, multi-local,  Virasoro representations on $\mathit{F^{\ten \frac{1}{2}}}$, we want to introduce the following Corollary of Proposition \ref{prop:heisdecomp} utilizing the combination of the new Heisenberg representation on $\mathit{F^{\ten \frac{1}{2}}}$ with the well known charge $c=\frac{1}{2}$ representation of Virasoro.
\begin{cor}(to Proposition \ref{prop:heisdecomp})
Define the graded dimension (character) of the Fock space\\ $\mathit{F}=\mathit{F^{\ten \frac{1}{2}}}$ as
\[
ch \mathit{F} :=tr_{\mathit{F}}q^{L^{1/2}_0}z^{h_0}
\]
We have
\begin{align}
\label{eqn:jac1}
ch \mathit{F}&=\prod_{i=1}^{\infty} (1+zq^{2i-1+\frac{1}{2}})(1+z^{-1}q^{2i-2+\frac{1}{2}})\\
\label{eqn:jac2}
&=\frac{1}{\prod_{i=1}^{\infty} (1-q^{2i})}\sum_{n\in \mathbb{Z}} z^nq^{\frac{n}{2}}q^{n^2}
\end{align}
By comparing the two identities we get the Jacobi identity
\begin{equation}
\label{eqn:Jacid}
\prod_{i=1}^{\infty} (1-q^{2i})(1+zq^{2i-\frac{1}{2}})(1+z^{-1}q^{2i-\frac{3}{2}})=\sum_{m\in \mathbb{Z}}z^m q^{\frac{m(2m+1)}{2}}.
\end{equation}
\end{cor}
\begin{remark}
For the boson-fermion correspondence of type A one gets a Jacobi identity (at $\lambda=\frac{1}{2}$ in the notation of \cite{Kac}, Chapter 5.1):
\begin{equation*}
\prod_{i=1}^{\infty} (1-q^i)(1-zq^{i-\frac{1}{2}})(1-z^{-1}q^{i-\frac{1}{2}})=\sum_{m\in \mathbb{Z}}z^m q^{m^2/2}
\end{equation*}
equivalent to  the Jacobi triple product identity (after exchanging $z$ to $-zq^{-\frac{1}{2}}$, see \cite{Kac}, Chapter 5.1):
\begin{equation}
\label{eqn:JacTriple}
\prod_{i=1}^{\infty} (1-q^i)(1-zq^{i-1})(1-z^{-1}q^{i})=\sum_{m\in \mathbb{Z}}(-1)^mz^m q^{m(m-1)/2}
\end{equation}
Here for the boson-fermion correspondence of type D-A we got \eqref{eqn:Jacid}, which curiously is equivalent to the same identity actually. In fact both of them are forms of the product identity for the same theta function $\theta_3$. In this case we need to exchange $z$ to $zq^{\frac{1}{2}}$ to get  the canonical theta function product form for $\theta_3$. And we need to exchange $z$ to $-zq^{\frac{3}{2}}$ and $q$ to $q^{\frac{1}{2}}$ to get the Jacobi triple product identity in the form \eqref{eqn:JacTriple}. This curiosity is of course explained by the previous Corollary \ref{cor:Heisequivalence}.
\end{remark}
\begin{proof}
From  \eqref{eqn:Vir1/2} and \eqref{eqn:normord} we have
\begin{align*}
L^{1/2}_0 &=\frac{1}{2}\sum _{m\in \mathbb{Z}+1/2} (-m-1/2):\phi^D_m \phi^D_{-m}:= \frac{1}{2}\sum _{n\in \mathbb{Z}} n:\phi^D_{-n-\frac{1}{2}} \phi^D_{n+\frac{1}{2}}:\\
&=\sum _{m\in \mathbb{Z}+1/2, m\geq 0} m:\phi^D_{-m} \phi^D_{m}:=\sum _{n\in \mathbb{Z}, n\geq 0} (n+\frac{1}{2}):\phi^D_{-n-\frac{1}{2}} \phi^D_{n+\frac{1}{2}}:
\end{align*}
So we first use the decomposition of  $\mathit{F^{\ten \frac{1}{2}}}$ by Length $\tilde{L}$.
For each $v=\phi^D_{-n_k-\frac{1}{2}}\dots \phi^D_{-n_2-\frac{1}{2}}\phi^D_{-n_1-\frac{1}{2}}|0\rangle$ from $\mathit{F_{(n)}^{\ten \frac{1}{2}}}$  we have (direct observation, but can be found in e.g. \cite{Triality})
\[
L^{1/2}_0 v=\left((n_1+\frac{1}{2})+\dots (n_k+\frac{1}{2})\right)v
\]
Now observe that on each $v=\phi^D_{-n_k-\frac{1}{2}}\dots\phi^D_{-n_2-\frac{1}{2}}\phi^D_{-n_1-\frac{1}{2}}|0\rangle$, due to the Heisenberg  $dg$ decomposition (Proposition \ref{prop:heisdecomp})
\[
h_0 v=\left(\#\{i=1, 2, \dots, k|n_i=\text{odd}\} -\#\{i=1, 2, \dots, k|n_i=\text{even}\}\right)v
\]
That means that we are multiplying by $z$ for every odd $n_i$, and by $z^{-1}$ for every $n_i$ that is even.
Hence \eqref{eqn:jac1} follows directly.

On the other hand, lets use the  decomposition of $\mathit{F^{\ten \frac{1}{2}}}$ into the Heisenberg  $dg$ decomposition (Proposition \ref{prop:heisdecomp}) $\mathit{F_{(n)}^{\ten \frac{1}{2}}}$, followed by the decomposition of each $\mathit{F_{(n)}^{\ten \frac{1}{2}}}$ into $\mathit{F_{(n, k)}^{\ten \frac{1}{2}}}$ by $deg_h$. A basis for $\mathit{F_{(n, k)}^{\ten \frac{1}{2}}}$ is given by the elements $v_{h, n, k}=h_{-k_{l}}\dots h_{-k_{1}}v_n$ with indecies varying with partitions $0<k_1\leq \dots \leq k_l$ of $k=k_1+\dots +k_l$.
For such elements we have of course $h_0 v_{h, n, k} =n v_{h, n, k}$, since they are in $\mathit{F_{(n)}^{\ten \frac{1}{2}}}$.
First, we have
\begin{align*}
L^{1/2}_0 v_0& = L^{1/2}_0 |0\rangle =0\\
L^{1/2}_0 v_n &= L^{1/2}_0 \phi^D_{-2n+1-\frac{1}{2}}\dots \phi^D_{-3-\frac{1}{2}}\phi^D_{-1-\frac{1}{2}}|0\rangle \\
&\quad = \left((1+\frac{1}{2})+ (3+\frac{1}{2}) +\dots +(2n-1+\frac{1}{2})\right) v_n=(n^2+\frac{n}{2}) v_n, \quad \text{for}\ n>0;\\
L^{1/2}_0 v_{-n} &= L^{1/2}_0 \phi^D_{-2n+2-\frac{1}{2}}\dots \phi^D_{-2-\frac{1}{2}}\phi^D_{-\frac{1}{2}}|0\rangle \\
&\quad =\left((0+\frac{1}{2}) + (2+\frac{1}{2}) + (4+\frac{1}{2})+\dots +(2n-2+\frac{1}{2})\right) v_{-n} =(n^2-\frac{n}{2})v_{-n}, \quad \text{for}\ n>0.
\end{align*}
Hence
\begin{align*}
L^{1/2}_0 v_{h, 0, k}& =L^{1/2}_0 h_{-k_{l}}\dots h_{-k_{1}}v_n=\left(2k_1+\dots +2k_l\right)v_{h, 0, k};\\
L^{1/2}_0 v_{h, n, k}& =L^{1/2}_0 h_{-k_{l}}\dots h_{-k_{1}}v_n=\left(2k_1+\dots +2k_l +n^2+\frac{n}{2}\right)v_{h, n, k}, \quad \text{for}\ n>0;\\
L^{1/2}_0 v_{h, -n, k}& =L^{1/2}_0 h_{-k_{l}}\dots h_{-k_{1}}v_{-n}=\left(2k_1+\dots +2k_l +n^2 -\frac{n}{2}\right)v_{h, -n, k}, \quad \text{for}\ n>0.
\end{align*}
Now since there are $p(k)$ such elements for partitions $0<k_1\leq \dots \leq k_l$ of $k=k_1+\dots +k_l$, we have
\begin{align*}
ch \mathit{F} &=\sum_{k\in \mathbb{Z}, k\geq 0}p(k)q^{2k}+\sum_{n, k\in \mathbb{Z}_+}p(k)\left(z^nq^{2k+n^2+\frac{n}{2}} +z^{-n}q^{2k+n^2-\frac{n}{2}}\right)\\
& =\sum_{k\in \mathbb{Z}, k\geq 0}p(k)q^{2k}\cdot \left(1+\sum_{n\in \mathbb{Z}_+} \left(z^nq^{n^2+\frac{n}{2}} +z^{-n}q^{n^2-\frac{n}{2}}\right)\right)\\
& =\frac{1}{\prod_{i=1}^{\infty} (1-q^{2i})}\left(1+\sum_{n\in \mathbb{Z}_+} \left(z^nq^{n^2+\frac{n}{2}} +z^{-n}q^{n^2-\frac{n}{2}}\right)\right)\\
& =\frac{1}{\prod_{i=1}^{\infty} (1-q^{2i})}\left(1+\sum_{n\in \mathbb{Z}_+} \left(z^nq^{\frac{n}{2}}q^{n^2} +z^{-n}q^{-\frac{n}{2}}q^{n^2}\right)\right)\\
&=\frac{1}{\prod_{i=1}^{\infty} (1-q^{2i})}\sum_{n\in \mathbb{Z}} z^nq^{\frac{n}{2}}q^{n^2}.
\end{align*}
\end{proof}
Now we proceed to constructing new, multi-local,  Virasoro representations on $\mathit{F^{\ten \frac{1}{2}}}$.
As we commented above, we will consider the \textbf{twisted vertex algebra} structure with space of states  $\mathit{F^{\ten \frac{1}{2}}}$, but with space of fields the $N=2$ field descendent  space  $\mathbf{\mathfrak{FD}}\{ \phi ^D(z); 2 \} $  (a ramified double cover of  $\mathit{F^{\ten \frac{1}{2}}}$). This space of fields, and specifically the multi-locality at both $z=w$ and $z=-w$ allowed us to construct the Heisenberg field--- an element of $\mathbf{\mathfrak{FD}}\{ \phi ^D(z); 2 \} $ acting on $\mathit{F^{\ten \frac{1}{2}}}$. This Heisenberg field, Proposition \ref{prop:heisdecomp} and the oscillator construction ensure that the Fock space $\mathit{F^{\ten \frac{1}{2}}}$ "inherits" a 1-parameter  Virasoro representation similar to \eqref{eqn:Vir-typeA} from the Fock space
$\mathit{F^{\ten 1}}$, but with a little twist:
\begin{prop}\label{prop:VirasoroNew}
The field
\begin{equation}
\label{eqn:normal-order-L-D}
L^{1}(z^2): =\frac{1}{2z^2}:h^D (z)h^D (z):
\end{equation}
has only  even-indexed modes,    $L^{1}(z^2):=\sum _{n\in \mathbb{Z}} L^1_{n} (z^2)^{-n-2}$ and moreover its modes $L_n$ satisfy the Virasoro algebra commutation relations with central charge $c=1$:
\[
[L_m, L_n] =(m-n)L_{m+n} +\delta_{m, -n}\frac{(m^3-m)}{12}.
\]
Equivalently, the field $L^{1}(z^2)$ has OPE with itself given by:
\begin{equation}
\label{eqn:VirOPEsD}
L^{1}(z^2)L^{1}(w^2)\sim \frac{1/2}{(z^2-w^2)^4} + \frac{2L^{1}(w^2)}{(z^2-w^2)^2}+ \frac{\partial_{w^2}L^{1}(w^2)}{(z^2-w^2)}.
\end{equation}
Furthermore, the field
\begin{equation}
\label{eqn:normal-order-L-D-lambda}
L^{\lambda, b}(z^2): =L^{1}(z^2) + (\frac{1}{2} -\lambda )\frac{1}{2z^2}\partial_z h^D (z)-\frac{b}{z^3}h^D(z) +\frac{16b^2+2b(1-2\lambda) -3(1-2\lambda)^2}{32z^4}=\sum _{n\in \mathbb{Z}} L^{\lambda}_{n} (z^2)^{-n-2}
\end{equation}
is a Virasoro field for every $\lambda, b \in \mathbb{C}$ with central charge $-12\lambda^2+12\lambda -2$. If $\lambda =\frac{1}{2}, \ b=0$, $L^{\frac{1}{2}, 0}(z^2)=L^{1}(z^2)$.
\end{prop}
\begin{remark}
One should notice that in contrast to \eqref{eqn:Vir-typeA2}, where if $b=0$ the corrective term $\frac{b(b-2\lambda +1)}{2z^2}$ vanishes  regardless of $\lambda$; here if $b=0$ the corrective term is only zero if also $\lambda =\frac{1}{2}$. I.e., here there is a corrective $\frac{1}{z^4}$  term except when $L^{\lambda, b}(z^2)$ coincides with $L^{1}(z^2)$ for $b=0, \lambda =\frac{1}{2}$.
This  reflects the fact that the space of fields here, $ \mathbf{\mathfrak{FD}}\{ \phi ^D(z); 2 \}$, is a twisted ramified cover of the space of states $\mathit{F^{\ten \frac{1}{2}}}$ (\cite{AngTVA}). For this reason a proof is necessary, even though this 2-parameter family utilizes the Hesienberg field, its normal product and derivative.
 \end{remark}
\begin{proof}
From Wick's theorem we have
\begin{align*}
L^{1}(z^2)&L^{1}(w^2)=
\frac{1}{4}\frac{1}{z^2w^2}:h ^D(z) h ^D(z): :h ^D(w)h ^D(w):  \\
&\sim \frac{1}{z^2w^2}\frac{zw}{(z^2-w^2)^2} :h ^D(z)h ^D(w):   +\frac{1}{2}\frac{1}{z^2w^2}\frac{zw}{(z^2-w^2)^2}\frac{zw}{(z^2-w^2)^2}\\
&\sim \frac{1}{zw}\frac{1}{(z^2-w^2)^2} :h ^D(z)h ^D(w):  +\frac{1}{2}\frac{1}{(z^2-w^2)^4}\\
&\sim \frac{1}{w^5}\left(\frac{1}{z}-\frac{1}{2}\frac{1}{(z-w)} -\frac{1}{2}\frac{1}{(z+w)}+\frac{w}{4}\frac{1}{(z-w)^2}- \frac{w}{4}\frac{1}{(z+w)^2}\right):h ^D(z)h ^D(w):  +\frac{1}{2}\frac{1}{(z^2-w^2)^4}\\
&\sim \frac{1}{w^5}\left(-\frac{1}{2}\frac{1}{(z-w)} -\frac{1}{2}\frac{1}{(z+w)}+\frac{w}{4}\frac{1}{(z-w)^2}- \frac{w}{4}\frac{1}{(z+w)^2}\right):h ^D(z)h ^D(w):   +\frac{1}{2}\frac{1}{(z^2-w^2)^4}
\end{align*}
Now we apply Taylor expansion formula Lemma \ref{lem:normalprodexpansion}:
\begin{align*}
L^{1}(z^2)L^{1}(w^2)&\sim -\frac{1}{2w^5}\frac{:h ^D(w)h ^D(w):}{(z-w)} -\frac{1}{2w^5}\frac{:h ^D(-w)h ^D(w):}{(z+w)}\\
&\quad +\frac{1}{4w^4}\frac{:h ^D(w)h ^D(w):}{(z-w)^2}- \frac{1}{4w^4}\frac{:h ^D(-w)h ^D(w):}{(z+w)^2}\\
&\quad +\frac{1}{4w^4}\frac{:\partial_w h ^D(w)h ^D(w):}{(z-w)}- \frac{1}{4w^4}\frac{:\partial_{-w}h ^D(-w)h ^D(w):}{(z+w)} +\frac{1}{2}\frac{1}{(z^2-w^2)^4}
\end{align*}
From Proposition \ref{prop:Heis} we have $h ^D(-w)=-h ^D(w)$, thus
\begin{align*}
L^{1}(z^2)&L^{1}(w^2)\sim -\frac{1}{2w^5}\frac{:h ^D(w)h ^D(w):}{(z-w)} +\frac{1}{2w^5}\frac{:h ^D(w)h ^D(w):}{(z+w)}\\
&\quad +\frac{1}{4w^4}\frac{:h ^D(w)h ^D(w):}{(z-w)^2}+ \frac{1}{4w^4}\frac{:h ^D(w)h ^D(w):}{(z+w)^2}\\
&\quad +\frac{1}{4w^4}\frac{:\partial_w h ^D(w)h ^D(w):}{(z-w)}- \frac{1}{4w^4}\frac{:\partial_w h ^D(w)h ^D(w):}{(z+w)} +\frac{1}{2}\frac{1}{(z^2-w^2)^4}\\
&\sim \frac{1}{2}\frac{1}{(z^2-w^2)^4} +\frac{(z^2+w^2)}{w^2}\frac{L^{1}(w^2)}{(z^2-w^2)^2} -\frac{1}{w^4}\frac{:h ^D(w)h ^D(w):}{(z^2-w^2)}  +\frac{1}{2w^3}\frac{:\partial_w h ^D(w)h ^D(w):}{(z^2-w^2)}\\
&\sim \frac{1}{2}\frac{1}{(z^2-w^2)^4} +\frac{(z^2+w^2)}{w^2}\frac{L^{1}(w^2)}{(z^2-w^2)^2}+\frac{1}{(z^2-w^2)}\frac{L^{1}(w^2)}{w^2}\\
&\quad +\frac{1}{(z^2-w^2)}\left(-\frac{:h ^D(w)h ^D(w):}{w^4}  +\frac{:\partial_w h ^D(w)h ^D(w):}{2w^3}\right)\\
&\sim \frac{1}{2}\frac{1}{(z^2-w^2)^4} +\frac{2L^{1}(w^2)}{(z^2-w^2)^2} +\frac{1}{(z^2-w^2)}\left(:h ^D(w)h ^D(w):\partial_{w^2}\frac{1}{2w^2}  +\frac{1}{2w^2}\partial_{w^2}:h ^D(w)h ^D(w):\right)\\
&\sim \frac{1}{2}\frac{1}{(z^2-w^2)^4} +\frac{2L^{1}(w^2)}{(z^2-w^2)^2} +\frac{\partial_{w^2}L^{1}(w^2)}{(z^2-w^2)} \\
\end{align*}
This proves that $L^{1}(z^2)=\frac{1}{2z^2}:h^D (z)h^D (z):$ is a Virasoro field with central charge $1$.

To calculate the OPEs of $L^{\lambda, b}(z^2): =\frac{1}{z^2}\left(\frac{1}{2}:h^D(z)^2: + K\partial_z h^D (z)-\frac{b}{z}h^D(z) +\frac{M}{z^2}\right)$, where we denote $K=\frac{1}{4}-\frac{\lambda}{2}$ and $M=\frac{16b^2+2b(1-2\lambda) -3(1-2\lambda)^2}{32}$, we need several OPEs, for which we use Wick's Theorem:
\begin{equation*}
\frac{1}{z^3}h^D (z)\frac{1}{w^3}h^D (w)\sim \frac{1}{z^3w^3}\frac{zw}{(z^2-w^2)^2}\sim  \frac{1}{z^2w^2}\frac{1}{(z^2-w^2)^2}
\sim \frac{1}{w^4}\frac{1}{(z^2-w^2)^2}-\frac{1}{w^6}\frac{1}{z^2-w^2}
\end{equation*}
\begin{align*}
\frac{1}{z^3}h^D (z)\frac{1}{w^2}\partial_wh^D (w)&\sim \frac{1}{2z^3w^2}\left(\frac{1}{(z-w)^3}+\frac{1}{(z+w)^3}\right)
\sim  \frac{1}{z^2w^2}\frac{z^2+3w^2}{(z^2-w^2)^3}\\
&\sim \frac{4}{w^2}\frac{1}{(z^2-w^2)^3}-\frac{3}{w^4}\frac{1}{(z^2-w^2)^2}+\frac{3}{w^6}\frac{1}{z^2-w^2};\\
\frac{1}{z^2}\partial_zh^D (z)\frac{1}{w^3}h^D (w)&\sim \frac{1}{2z^3w^2}\left(\frac{-1}{(z-w)^3}+\frac{1}{(z+w)^3}\right)
\sim  \frac{-1}{z^2w^2}\frac{3z^2+w^2}{(z^2-w^2)^3}\\
&\sim -\frac{4}{w^2}\frac{1}{(z^2-w^2)^3}+\frac{1}{w^4}\frac{1}{(z^2-w^2)^2}-\frac{1}{w^6}\frac{1}{z^2-w^2};
\end{align*}
Thus
\begin{equation*}
\frac{1}{z^2}\partial_zh^D (z)\frac{1}{w^3}h^D (w) +\frac{1}{z^3}h^D (z)\frac{1}{w^2}\partial_wh^D (w)\sim -\frac{2}{w^4}\frac{1}{(z^2-w^2)^2}+\frac{2}{w^6}\frac{1}{z^2-w^2}
\end{equation*}
Next,
\begin{align*}
\frac{1}{z^2}\partial_z h^D (z)\frac{1}{w^2}\partial_w h^D (w)&\sim \frac{1}{4z^2w^2}\partial_z \partial_w \left(\frac{1}{(z-w)^2} - \frac{1}{(z+w)^2}\right)\\
&\sim -\frac{1}{4z^2w^2}\left(\frac{6}{(z-w)^4} + \frac{6}{(z+w)^4}\right)= \frac{-3}{z^2 w^2(z^2-w^2)^2} +\frac{-24}{(z^2-w^2)^4}\\
&\sim \frac{-24}{(z^2-w^2)^4}- \frac{3}{w^4(z^2-w^2)^2} +\frac{3}{w^6(z^2-w^2)}.
\end{align*}
\begin{align*}
\frac{1}{z^3}h^D (z)L^{1}(w^2)&= \frac{1}{2z^3w^2}h^D(z) :h^D (w)h^D (w):\sim\frac{1}{z^3w^2}\frac{zw}{(z^2-w^2)^2}h^D (w)\\
&\sim \frac{1}{(z^2-w^2)^2}\frac{h^D(w)}{w^3}-\frac{1}{z^2-w^2}\frac{h^D(w)}{w^5}.
\end{align*}
\begin{align*}
L^{1}(z^2)\frac{1}{w^3}h^D (w)&= \frac{1}{2z^2w^3} :h^D (z)h^D (z):h^D(w)\sim\frac{1}{4z^2w^3}\left(\frac{1}{(z-w)^2} -\frac{1}{(z+w)^2}\right)h^D (z)\\
&\sim \left(\frac{1}{4w^5} -\frac{(z-w)}{2w^6}+\dots\right)\left(h^D(w) +(z-w)\partial_wh^D(w)+\dots\right)\frac{1}{(z-w)^2}\\
&\quad  -
\left(\frac{1}{4w^5} +\frac{(z+w)}{2w^6}+\dots\right)\left(h^D(-w) +(z+w)\partial_{-w}h^D(-w)+\dots\right)\frac{1}{(z+w)^2}\\
&\sim \frac{z^2 +w^2}{2(z^2-w^2)^2}\frac{h^D(w)}{w^5} -\frac{1}{z^2-w^2}\frac{h^D(w)}{w^5} +\frac{1}{2w^4}\frac{1}{z^2-w^2} \partial_wh^D(w)\\
&\sim \frac{1}{(z^2-w^2)^2}\frac{h^D(w)}{w^3} -\frac{1}{2(z^2-w^2)}\frac{h^D(w)}{w^5} +\frac{1}{2w^4}\frac{1}{z^2-w^2} \partial_wh^D(w);
\end{align*}
where we used that $h^D (-w)=-h^D (w)$,  and so we get
\begin{equation*}
L^{1}(z^2)\frac{1}{w^3}h^D (w)+\frac{1}{z^3}h^D (z)L^{1}(w^2)\sim \frac{2}{(z^2-w^2)^2}\frac{h^D(w)}{w^3} -\frac{3}{2(z^2-w^2)}\frac{h^D(w)}{w^5} +\frac{1}{2w^4}\frac{1}{z^2-w^2} \partial_wh^D(w).
\end{equation*}
Further OPEs:
\begin{align*}
\frac{1}{z^2}\partial_z h^D (z)L^{1}(w^2)&= \frac{1}{2z^2w^2}\partial_z h^D(z) :h^D (w)h^D (w):\sim\frac{1}{2z^2w^2}\left(-\frac{1}{(z-w)^3} +\frac{1}{(z+w)^3}\right)h^D (w)\\
L^{1}(z^2)\frac{1}{w^2}\partial_w h^D (w)&= \frac{1}{2z^2w^2}:h^D (z)h^D (z):\partial_w h^D(w) \sim\frac{1}{2z^2w^2}\left(\frac{1}{(z-w)^3} +\frac{1}{(z+w)^3}\right)h^D (z)\\
&\quad \sim \frac{1}{2z^2w^2}\left(\frac{h^D (w)}{(z-w)^3} + \frac{\partial_w h^D (w)}{(z-w)^2} + \frac{\partial^{(2)}_w h^D (w)}{z-w}\right)\\
&\quad \quad + \frac{1}{2z^2w^2}\left(\frac{h^D (-w)}{(z+w)^3} + \frac{\partial_{-w} h^D (-w)}{(z+w)^2} + \frac{\partial^{(2)}_{-w} h^D (-w)}{z+w}\right);
\end{align*}
\begin{align*}
\Big(L^{1}(z^2)\frac{1}{w^2}\partial_w h^D (w) &+ \frac{1}{z^2}\partial_z h^D (z)L^{1}(w^2)\Big) \sim
 \frac{1}{2z^2w^2}\left(\frac{\partial_w h^D (w)}{(z-w)^2} + \frac{\partial_{w} h^D (w)}{(z+w)^2}\right) \\
&\quad \quad \hspace{3cm} + \frac{1}{2z^2w^2}\left(\frac{\partial^{(2)}_w h^D (w)}{z-w}-\frac{\partial^{(2)}_{w} h^D (w)}{z+w}\right)\\
&\sim \left(-\frac{1}{z^2w^2}\frac{1}{z^2-w^2}+\frac{1}{w^2} \frac{2}{(z^2-w^2)^2}\right)\partial_w h^D (w)  + \frac{1}{2z^2w}\frac{1}{z^2-w^2}\partial^{2}_w h^D (w)\\
&\sim \left(-\frac{1}{w^2}\frac{1}{z^2-w^2}+\frac{2}{(z^2-w^2)^2}\right)\frac{1}{w^2}\partial_w h^D (w)  + \frac{1}{2w^3}\frac{1}{z^2-w^2}\partial^{2}_w h^D (w)\\
&\sim \frac{2}{(z^2-w^2)^2}\frac{1}{w^2}\partial_w h^D (w)  + \frac{1}{z^2-w^2}\partial_{w^2}\left(\frac{1}{w^2}\partial_w h^D (w)\right)
\end{align*}
Finally
\begin{align*}
\Big(L^{1}(z^2)&+ K\frac{1}{z^2}\partial_z h^D (z)-\frac{b}{z^3}h^D(z) +\frac{M}{z^4}\Big) \Big(L^{1}(w^2)+K\frac{1}{w^2}\partial_w h^D (w)-\frac{b}{w^3}h^D(z) +\frac{M}{w^4}\Big)\\
&\sim \frac{1}{2}\frac{1}{(z^2-w^2)^4} +\frac{2L^{1}(w^2)}{(z^2-w^2)^2} +\frac{\partial_{w^2}L^{1}(w^2)}{(z^2-w^2)}  +\frac{-24K^2}{(z^2-w^2)^4}- \frac{3K^2}{w^4(z^2-w^2)^2} +\frac{3K^2}{w^6(z^2-w^2)}\\
&\quad +\frac{b^2}{w^4}\frac{1}{(z^2-w^2)^2}-\frac{b^2}{w^6}\frac{1}{z^2-w^2} +\frac{2Kb}{w^4}\frac{1}{(z^2-w^2)^2}-\frac{2Kb}{w^6}\frac{1}{z^2-w^2}\\
&\quad +\frac{2}{(z^2-w^2)^2}K\frac{1}{w^2}\partial_w h^D (w)  + \frac{1}{z^2-w^2}\partial_{w^2}\left(K\frac{1}{w^2}\partial_w h^D (w)\right)\\
&\quad -\frac{2b}{(z^2-w^2)^2}\frac{h^D(w)}{w^3} +\frac{3b}{2(z^2-w^2)}\frac{h^D(w)}{w^5} -\frac{b}{2w^4}\frac{1}{z^2-w^2} \partial_wh^D(w)\\
&\sim \frac{1}{2}\frac{1-48K^2}{(z^2-w^2)^4} +\frac{1}{(z^2-w^2)^2}\left(2L^{1}(w^2)+2K\frac{1}{w^2}\partial_w h^D (w) -2\frac{b}{w^3}h^D(w)+\frac{b^2+2Kb-3K^2}{w^4}\right) \\
&\quad +\frac{1}{z^2-w^2}\left(\partial_{w^2}L^{1}(w^2) + \partial_{w^2}\big(K\frac{1}{w^2}\partial_w h^D (w)\big) -\partial_{w^2}\big(\frac{b}{w^3} h^D (w)\big) +\partial_{w^2}\frac{b^2+2Kb-3K^2}{2w^4}\right)
\end{align*}
\end{proof}
We now proceed with  the existence of  other Virasoro representations on $\mathit{F^{\ten \frac{1}{2}}}$.
The 2-point-locality allows us to construct in a general standard way two new Virasoro fields from an existing 1-point Virasoro field.
\begin{prop}\label{prop:doublingcharge}
\textbf{I.} Let
$L^c(z): =\sum _{n\in \mathbb{Z}} L^c_{n} z^{-n-2}$
is a Virasoro 1-point local field with central charge $c$ and  OPE:
\begin{equation}
\label{eqn:VirOPEsc}
L^c(z)L^c(w)\sim \frac{c/2}{(z-w)^4} + \frac{2L^c(w)}{(z-w)^2}+ \frac{\partial_{w}L^c(w)}{(z-w)},
\end{equation}
then $\tilde{L}^c(z):= L^c(-z) =\sum _{n\in \mathbb{Z}} L^c_{n} (-z)^{-n-2}$ is also a Virasoro 1-point field with central charge $c$.\\
\textbf{II.} Let  a vector space $V$ be a module for the Virasoro algebra $Vir$ with generators  $L_n$ and a central charge $c$ via the representation $r_c: Vir\to End(V), \quad  L_n\mapsto  r_c(L_n)$. Let $N$ be a positive integer, and let $\epsilon$ be a primitive $N$th root of unity. Then $V$ is also a  module for the Virasoro algebra $Vir$ with central charge $Nc$ via the representation
\begin{align*}
r_{Nc}: Vir&\to End(V)\\
L_n&\mapsto  \frac{1}{N}r_c(L_{Nn})+\frac{N^2-1}{24N}c\delta_{n, 0}.
\end{align*}
Equivalently, if $L^c(z)$ is a Virasoro 1-point local field,
then
\begin{equation}
L^{Nc}(z^N): =\frac{1}{N^2z^{2N}}\left(\sum_{i=0}^{N-1}L^{c, \text{shifted}}(\epsilon ^i z)\right)+\frac{N^2-1}{24Nz^{2N}}c=\sum _{n\in \mathbb{Z}} L^{Nc}_{n} (z^N)^{-n-2},
\end{equation}
where $L^{c, \text{shifted}} (z): =z^2L^c(z): =\sum _{n\in \mathbb{Z}} L^c_{n} z^{-n}$,
is a Virasoro N-point local field with OPE:
\begin{equation}
\label{eqn:VirOPEsNc}
L^{Nc}(z^N)L^{Nc}(w^N)\sim \frac{Nc/2}{(z^N-w^N)^4} + \frac{2L^{Nc}(w^N)}{(z^N-w^N)^2}+ \frac{\partial_{w^N}L^{Nc}(w^N)}{(z^N-w^N)}.
\end{equation}
\end{prop}
{\it Proof:}
For part \textbf{I.}, a direct calculation with the OPEs is trivial:
\begin{equation*}
\tilde{L}^c(z)\tilde{L}^c(w)=L^c(-z)L^c(-w)\sim \frac{c/2}{(z-w)^4} + \frac{2L^c(-w)}{(z-w)^2}+ \frac{-\partial_{-w}L^c(-w)}{(z-w)}\sim \frac{c/2}{(z-w)^4} + \frac{2\tilde{L}^c(w)}{(z-w)^2}+ \frac{\partial_{w}\tilde{L}^c(w)}{(z-w)}.
\end{equation*}
 For part \textbf{II.}, we couldn't find a reference to this rather obvious (in hindsight) fact in the standard books and papers on  representation of the Virasoro algebra, so we prove it here.
It is easier to prove by using modes directly. We need to calculate the commutator $[L^{Nc}_{m}, L^{Nc}_{n}]$. We consider the case $m+n\neq 0$ first. Then we have
\[
[L^{Nc}_{m}, L^{Nc}_{n}]=  \frac{1}{N^2}[L^{c}_{m}, L^{c}_{n}]=\frac{1}{N^2}(Nm-Nn)L^{c}_{Nm+Nn}=(m-n)\frac{L^{c}_{Nm+Nn}}{N}=L^{Nc}_{m+n}
\]
For the case $m+n= 0$, we can assume $m\neq 0$, hence
\begin{align*}
[L^{Nc}_{m}, L^{Nc}_{n}]&=  \frac{1}{N^2}[L^{c}_{m}, L^{c}_{n}]=\frac{1}{N^2}(Nm-Nn)L^{c}_{0}
+\frac{1}{N^2}\frac{N^3m^3-Nm}{12}c\\
&=(m-n)\frac{L^{c}_{0}}{N} +\frac{1}{12}\left(m^3-\frac{m}{N^2}\right)Nc=2mL^{Nc}_{0}-2m\frac{N^2-1}{24N}c +\frac{1}{12}\left(m^3-\frac{m}{N^2}\right)Nc\\
&=2mL^{Nc}_{0} +\frac{m^3-m}{12}Nc. \hspace{8.5cm} \ \square
\end{align*}
In the case of $\mathit{F^{\ten \frac{1}{2}}}$ the well-known Virasoro field of central charge  $c=\frac{1}{2}$ allows us to construct immediately another Virasoro field with central charge  $c=\frac{1}{2}$:
\[
\tilde{L}^{1/2}(z)=\frac{1}{2}:\partial_{-z} \phi ^D (-z) \phi ^D(-z):\in \mathbf{\mathfrak{FD}}\{ \phi ^D(z); 2 \}.
\]
$\tilde{L}^{1/2}(z)$ is a 1-point local Virasoro field with central charge $c=\frac{1}{2}$. We want to note that even though the field $\tilde{L}^{1/2}(z)$ is 1-point self local, it is 2-point  mutually local with the other fields in $\mathbf{\mathfrak{FD}}\{ \phi ^D(z); 2 \}$, including with $L^{1/2}(z)$. This  is the reason it can't be considered as part of the super vertex algebra structure on  $\mathit{F^{\ten \frac{1}{2}}}$, but as part of the twisted vertex algebra structure.

Moreover, $L^{1/2}(z)$ and $\tilde{L}^{1/2}(z)$ are the only 1-point local Virasoro fields on $\mathit{F^{\ten \frac{1}{2}}}$.
To prove that $L^{1/2}(z)$ and $\tilde{L}^{1/2}(z)$ are the only 1-point local Virasoro fields on $\mathit{F^{\ten \frac{1}{2}}}$ one needs to consider all weight two multi-local fields on  $\mathit{F^{\ten \frac{1}{2}}}$. Besides $:h^D (z)h^D (z):$ and $\partial_zh^D (z)$ and their shifts that we already considered (they produced a family of two-point Virasoro  fields), there are 4 fields to consider:
\begin{equation}
:\partial_{z} \phi ^D (z) \phi ^D(z):,\  :\partial_{-z} \phi ^D (-z) \phi ^D(-z):, \ :\partial_{z} \phi ^D (z) \phi ^D(-z):,\  :\partial_{-z} \phi ^D (-z) \phi ^D(z):.
\end{equation}
These 4 fields are linearly independent weight two fields, but moreover they are the only fields that have at most 4th order poles in their OPEs. In order to prove that $L^{1/2}(z)$ and $\tilde{L}^{1/2}(z)$ are the only 1-point local Virasoro fields one considers arbitrary linear combinations of the fields above. We omit this rather long, but straightforward calculation,  especially since  linear combination of two of them appears in $:h^D(w)^2:$, as we saw in  \eqref{eqn:h-normorder},  and  also linear combination of another two of them gives $\partial_zh^D (z)$, \eqref{eqn:hderiv}.
Direct calculation shows  there is no family of 1-point Virasoro fields.

 The same two fields  $L^{1/2}(z)$ and $\tilde{L}^{1/2}(z)$, but in a linear combination will also give the 2-point Virasoro field as an application of the Part \textbf{II.} of the Proposition  \ref{prop:doublingcharge}:
The field
\begin{equation}
\label{eqn:normal-order-L-D--1}
\tilde{L}^{1}(z^2): =\frac{1}{4z^2}\left(L^{1/2}(z) +\tilde{L}^{1/2}(z)\right)+\frac{1}{32z^4}=\frac{1}{8z^2}\left(:\partial_{z} \phi ^D (z) \phi ^D(z): +:\partial_{-z} \phi ^D (-z) \phi ^D(-z):\right)+\frac{1}{32z^4}
\end{equation}
has only  even-indexed modes,    $\tilde{L}^{1}(z^2):=\sum _{n\in \mathbb{Z}} \tilde{L}^{1}_{n} (z^2)^{-n-2}$ and moreover its modes $\tilde{L}^{1}_n$ satisfy the Virasoro algebra commutation relations with central charge $c=1$.
 Since $\tilde{L}^{1/2}(z)=L^{1/2}(-z)$, the formula above is   direct application of Part \textbf{II.} of Proposition \ref{prop:doublingcharge} for $N=2$. It also explains the constant $\frac{1}{32}$, as  $c=\frac{1}{2}$, and thus $\frac{N^2-1}{24N}c=\frac{1}{32}$. The constant $\frac{1}{32}$ seemed rather unusual in terms of Virasoro representations--one can only think of it as splitting the constant $\frac{1}{16}$ that appears in the twisted oscillator representation (although the Virasoro representation  given by $\tilde{L}^{1}(z^2)$ is not an oscillator representation). Nevertheless this same constant will appear again in the new  Virasoro representation for the correspondence of type C (\cite{AngVirC}). Here we have for the modes of the  field $\tilde{L}^{1}(z^2)$
\begin{equation}
\tilde{L}^{1}_n =\frac{1}{2}L^{1/2}_{2n} +\frac{1}{32}\delta_{n, 0}.
\end{equation}
Such a relation doesn't hold in the case of type C (\cite{AngVirC}), if for no other reason that there is no "half-charge" field there, nevertheless similar "double the charge, half the modes" relation holds in general, as we proved in  part \textbf{II.} of Proposition \ref{prop:doublingcharge} for $N=2$.

Note that  in this case  the field  $\tilde{L}^{1}(z^2)$, although a Virasoro field with central charge 1, and linearly independent from $L^{1}(z^2)$, doesn't contribute much new information, since
\[
\tilde{L}^{1}(z^2)=L^{\frac{1}{2}, -\frac{1}{4}}(z^2)=L^{1}(z^2) +\frac{1}{4z^3}h^D(z) +\frac{1}{32z^4},
\]
i.e., $\tilde{L}^{1}(z^2)$ is a an element of the two parameter family $L^{\lambda, b}(z^2)$.
\begin{remark}
The boson-fermion correspondence of type D-A  was discussed  during the work on both \cite{AngTVA}  and on  \cite{Rehren}, where it was considered  from the point of view of   "multilocal fermionization". The author thanks K. Rehren for the helpful discussions and a physics point of view on the subject.  Some of the Virasoro fields we derive also appear independently in \cite{Rehren},  as stress-energy tensors.
\end{remark}

      To summarize the various Virasoro fields on $\mathit{F^{\ten \frac{1}{2}}}$ we obtained. The two   1-point-local Virasoro fields, $L^{1/2}(z)$ and $\tilde{L}^{1/2}(z)$ are  due to the  super vertex algebra structure on each of the two "sheets" (each isomorphic to $\mathit{F^{\ten \frac{1}{2}}}$) that the  twisted vertex algebra structure  "glues" together. The  2-point local Virasoro fields  $L^{1}(z^2)$, $\tilde{L}^{1}(z^2)$ and $L^{\lambda}(z^2)$ are due to the overall twisted vertex algebra structure  with space of fields $ \mathbf{\mathfrak{FD}}\{ \phi ^D(z); 2 \}$ that is  responsible for the  bosonization of type D-A. We expect that there is a genuine (non-splitting) representation of a version of a two-point Virasoro algebra (see e.g. \cite{Krich-Nov1}, \cite{Krich-Nov2}, \cite{Schl}, \cite{Cox-Vir}) arising from a linear combination of the 4  weight-two fields \begin{equation}
:\partial_{z} \phi ^D (z) \phi ^D(z):,\  :\partial_{-z} \phi ^D (-z) \phi ^D(-z):, \ :\partial_{z} \phi ^D (z) \phi ^D(-z):,\  :\partial_{-z} \phi ^D (-z) \phi ^D(z):.
\end{equation}
But although the fields $L^{1}(z^2)$, $\tilde{L}^{1}(z^2)$ and $L^{\lambda}(z^2)$ appear to be 2-point-local, they can be reduced to local fields after change of variables, i.e., these Virasoro fields do produce  genuine representations of the Virasoro algebra on $\mathit{F^{\ten \frac{1}{2}}}$ (as opposed to representations of a 2-point Virasoro algebra). We summarize these results in the following:
\begin{thm}\label{thm:summaryVir}
\textbf{I.} Let $L^{1/2}(z)=\frac{1}{2}:\partial_{z} \phi ^D (z) \phi ^D(z):$, $\tilde{L}^{1/2}(z)=\frac{1}{2}:\partial_{-z} \phi ^D (-z) \phi ^D(-z):$. The 1-point-local fields $L^{1/2}(z)$ and $\tilde{L}^{1/2}(z)$ are  Virasoro fields with central charge $\frac{1}{2}$, i.e., the modes of each of these two fields generate a representation of the Virasoro algebra with central charge $\frac{1}{2}$ on $\mathit{F^{\ten \frac{1}{2}}}$.\\
\textbf{II.} Let $L^{1}(z^2): =\frac{1}{2z^2}:h^D (z)h^D (z):$,
$\tilde{L}^{1}(z^2):=\frac{1}{8z^2}\left(:\partial_{z} \phi ^D (z) \phi ^D(z): +:\partial_{-z} \phi ^D (-z) \phi ^D(-z):\right)+\frac{1}{32z^4}$.
 The 2-point-local fields $L^{1}(z^2)$ and $\tilde{L}^{1}(z^2)$ are  Virasoro fields  with central charge $1$, and their modes generate respectively representations of the Virasoro algebra with central charge $1$ on $\mathit{F^{\ten \frac{1}{2}}}$.\\
\textbf{III.} Let $L^{\lambda, b}(z^2): =L^{1}(z^2) + (\frac{1}{2} -\lambda )\frac{1}{2z^2}\partial_z h^D (z)-\frac{b}{z^3}h^D(z) +\frac{16b^2+2b(1-2\lambda) -3(1-2\lambda)^2}{32z^4}$, $\lambda, b \in \mathbb{C}$. For each $\lambda, b \in \mathbb{C}$,\  $L^{\lambda, b}(z^2)$ is a 2-point-local Virasoro field acting on $\mathit{F^{\ten \frac{1}{2}}}$, whose modes generate a representation of the Virasoro algebra with central charge $-2+12\lambda -12\lambda^2$. If $\lambda =\frac{1}{2}, \ b=0$, $L^{\frac{1}{2}, 0}(z^2)=L^{1}(z^2)$; If $\lambda =\frac{1}{2}, \ b=\frac{-1}{4}$, $L^{\frac{1}{2}, -\frac{1}{4}}(z^2)=\tilde{L}^{1}(z^2)$.
\end{thm}

Starting with Igor Frenkel's work in \cite{Frenkel-BF}, and later, the boson-fermion correspondence of type A   was related to various kinds of Howe-type dualities (see e.g. \cite{WangKac}, \cite{WangDuality}, \cite{WangDual}), and again there was an attempt to understand a boson-fermion correspondence as being identified by, or equivalent to,  a certain duality or a basic  modules decomposition (see e.g. \cite{WangDuality}, \cite{WangDual}).
 For instance, in \cite{WangDual}  it was mentioned that   "the $(GL_1, \widehat{\mathcal{D}})$ -duality in Theorem 5.3 is essentially the celebrated boson-fermion correspondence".  The  $(GL_1,  \widehat{\mathcal{D}})$ -duality of Theorem 5.3  of \cite{WangDual} in the case of $\mathit{F^{\ten 1}}$, as Wang himself mentions, is equivalent to \eqref{Aequiv}.
 We want to show that $\mathit{F^{\ten 1}}\cong\mathit{F^{\ten \frac{1}{2}}}$ also as $\widehat{\mathcal{D}}$ modules. $\widehat{\mathcal{D}}$ denotes the universal central extension of the Lie algebra of differential operators on the circle, often denoted also by $W_{1+\infty}$, see e.g.,  \cite{Kac}, \cite{WangDual}, which we briefly recall here.  Denote by $\mathcal{D}=D\!i\!f\!f  \mathbb{C}^*$ the Lie algebra of differential operators
\[
\sum_{j=0}^M a_j(t)\partial_t^j, \quad a_j(t)\in \mathbb{C}[t, t^{-1}]
\]
Denote also
\[
J^k_n=t^{n+k} (-\partial_t)^k
\]
The Lie algebra $\mathcal{D}$ has a basis consisting of the operators $J_n^k$.  Moreover it can be realized as a subalgebra of the $\tilde{gl}_\infty$ via
\[
J^k_n\mapsto (-1)^k{{-j}\choose {k}}E_{j-n, j}
\]
Let $\widehat{\mathcal{D}}$ denotes the universal central extension of the Lie algebra $\mathcal{D}$, $\widehat{\mathcal{D}}=\mathcal{D}+\mathbb{C}C$, and by abuse of notation we use the same notation $J^k_n$ for the corresponding elements in $\widehat{\mathcal{D}}$.  Consequently, $\widehat{\mathcal{D}}$ can be realized as a subalgebra of $a_\infty$. (Thus we already have a representation of $\widehat{\mathcal{D}}$ on $\mathit{F^{\ten 1}}\cong\mathit{F^{\ten \frac{1}{2}}}$, as a consequence of \cite{ACJ2}). Nevertheless,  introduce
the fields
\begin{equation}
J^{k}(z^2)=\sum_{n\in \mathbb{Z}} J^k_n z^{-2k-2n-2}
\end{equation}
\begin{prop}
The assignment
\[
J^k(z^2)\mapsto \frac{(-1)^k}{4z}:\left(\phi ^D(z)-\phi ^D(-z)\right)\partial_{z^2}^k \left(\phi ^D(z)+\phi ^D(-z)\right):
\]
gives an action of the Lie algebra  $\widehat{\mathcal{D}}$ on $\mathit{F^{\ten \frac{1}{2}}}$. Furthermore, each component $\mathit{F_{(m)}^{\ten \frac{1}{2}}}\in \mathit{F^{\ten \frac{1}{2}}}$ is irreducible under the action of $\widehat{\mathcal{D}}$, and  as $\widehat{\mathcal{D}}$ modules
\[
\mathit{F^{\ten 1}}\cong\mathit{F^{\ten \frac{1}{2}}}=\oplus_{m\in \mathbb{Z}}\mathit{F_{(m)}^{\ten \frac{1}{2}}}
\]
\end{prop}
\begin{proof}
One can check that the commutation relations are satisfied directly-- it easiest to use the fact that
\[
\frac{\phi ^D(z)-\phi ^D(-z)}{2z}\mapsto \psi^+(z^2), \quad \frac{\phi ^D(z)+\phi ^D(-z)}{2}\mapsto \psi^-(z^2)
\]
where the $\psi^+(z)$ and  $\psi^-(z)$ are the generating fields of the charged free fermions Fock space $\mathit{F^{\ten 1}}$. The fact that the fields $(-1)^k:\psi^+(z)\partial_z^k\psi^+(z):$ have OPEs satisfying the commutation relations for  $\widehat{\mathcal{D}}$ is well know, see e.g. \cite{Kac} Theorem 5.3. Further, we have
\begin{equation}
[J^0_m, J^0_n]=m\delta _{m+n,0}C
\end{equation}
Thus the Heisenberg algebra  $\mathcal{H}_{\mathbb{Z}}$ is a subalgebra of $\widehat{\mathcal{D}}$.
Note that we have
\begin{equation}
J^{0}(z^2)=\sum_{n\in \mathbb{Z}} J^k_n z^{-2n-2} \mapsto \ \frac{1}{4z}:\phi ^D(z)\phi ^D(-z):=\frac{1}{z}h^D(z), \quad
J^0_m \mapsto h^D_m
\end{equation}
Since $\mathit{F_{(m)}^{\ten \frac{1}{2}}}\in \mathit{F^{\ten \frac{1}{2}}}$ is irreducible under the action of the Heisenberg algebra  $\mathcal{H}_{\mathbb{Z}}$, that completes the proof.
\end{proof}
To summarize,  we proved that $\mathit{F^{\ten 1}}\cong\mathit{F^{\ten \frac{1}{2}}}$   as graded vector spaces,  as  $\mathcal{H}_{\mathbb{Z}}$  modules, and as $W_{1+\infty}$ modules, although obviously there should be some difference in terms of the physics structures on these two  spaces. As super vertex algebras, i.e., twisted vertex algebras of order $N=1$,  $\mathit{F^{\ten 1}}$ and  $\mathit{F^{\ten \frac{1}{2}}}$ are not equivalent. But if we introduce an $N=2$ twisted vertex algebra structures both on $\mathit{F^{\ten \frac{1}{2}}}$ and on $\mathit{F^{\ten 1}}$, there is an isomorphism of  twisted vertex algebra structures which constitutes the boson-fermion correspondence of type D-A. If we consider a twisted vertex algebra structure with $N>2$, then once more the twisted vertex algebra structures on  $\mathit{F^{\ten 1}}$ and  $\mathit{F^{\ten \frac{1}{2}}}$ will become inequivalent, as we will show in a follow-up paper. For $N>2$, we can indeed bosonize the Fock space $\mathit{F^{\ten \frac{1}{2}}}$, but the resulting bosonic twisted vertex algebra will be a rank one non-integer lattice chiral algebra, and thus not isomorphic to the rank one odd lattice vertex algebra  $\mathit{B_A}$ ($\mathit{B_A}$ is the super vertex algebra isomorphic to  $\mathit{F^{\ten 1}}$).
 This shows that the type of chiral algebra structure on $\mathit{F^{\ten 1}}$ vs $\mathit{F^{\ten \frac{1}{2}}}$ is of great importance, in particular the set of points of locality is a necessary part of the data describing any boson-fermion correspondence. Thus the chiral algebra  structure on $\mathit{F^{\ten \frac{1}{2}}}$, vs the chiral structure on $\mathit{F^{\ten 1}}$ is really what distinguishes these spaces and their physical structure and what defines a boson-fermion correspondence.

\def\cprime{$'$}


\begin{thebibliography}{DJKM81b}

\bibitem[ACJ13]{ACJ2}
Iana~I. Anguelova, Ben Cox, and Elizabeth Jurisich.
\newblock Representations of $a_{\infty}$ and $d_{\infty}$ with central charge
  1 on the {F}ock space $\mathit{F^{\otimes \frac{1}{2}}}$.
\newblock {\em Journal of Physics: Conference Series, 474(1)}, 474(1):20, 2013.

\bibitem[ACJ14]{ACJ}
Iana~I. Anguelova, Ben Cox, and Elizabeth Jurisich.
\newblock ${N}$-point locality for vertex operators: normal ordered products,
  operator product expansions, twisted vertex algebras.
\newblock {\em J. Pure Appl. Algebra}, 2014.
\newblock in press.

\bibitem[Ang13a]{Ang-Varna2}
Iana Anguelova.
\newblock Boson-fermion correspondence of type {B} and twisted vertex algebras.
\newblock In {\em Proceedings of the 9-th International Workshop "Lie Theory
  and Its Applications in Physics" (LT-9), Varna, Bulgaria}, Springer
  Proceedings in Mathematics and Statistics, 2013.

\bibitem[Ang13b]{AngTVA}
Iana~I. Anguelova.
\newblock Twisted vertex algebras, bicharacter construction and boson-fermion
  correspondences.
\newblock {\em J. Math. Phys.}, 54(12):38, 2013.

\bibitem[Ang14]{AngVirC}
Iana~I. Anguelova.
\newblock Virasoro structures in the twisted vertex algebra of the particle
  correspondence of type {C}.
\newblock to appear in the Proceedings of the 10-th International Workshop "Lie
  Theory and Its Applications in Physics" (LT-10), Varna, Bulgaria, 2014.

\bibitem[BS83]{MR85g:81096}
N.~N. Bogoliubov and D.~V. Shirkov.
\newblock {\em Quantum fields}.
\newblock Benjamin/Cummings Publishing Co. Inc. Advanced Book Program, Reading,
  MA, 1983.
\newblock Translated from the Russian by D. B. Pontecorvo.

\bibitem[CGLZ13]{Cox-Vir}
Ben Cox, Xiangqian Guo, Rencai Lu, and Kaiming Zhao.
\newblock {$N$-point Virasoro Algebras and Their Modules of Densities}.
\newblock 2013.

\bibitem[DJKM81a]{DJKM3}
Etsur{\=o} Date, Michio Jimbo, Masaki Kashiwara, and Tetsuji Miwa.
\newblock Transformation groups for soliton equations. {III}. {O}perator
  approach to the {K}adomtsev-{P}etviashvili equation.
\newblock {\em J. Phys. Soc. Japan}, 50(11):3806--3812, 1981.

\bibitem[DJKM81b]{DJKM6}
Etsur{\=o} Date, Michio Jimbo, Masaki Kashiwara, and Tetsuji Miwa.
\newblock Transformation groups for soliton equations. {VI}. {KP} hierarchies
  of orthogonal and symplectic type.
\newblock {\em J. Phys. Soc. Japan}, 50(11):3813--3818, 1981.

\bibitem[DKM81]{DJKM-1}
Etsur{\=o} Date, Masaki Kashiwara, and Tetsuji Miwa.
\newblock Transformation groups for soliton equations. {II}. {V}ertex operators
  and {$\tau $} functions.
\newblock {\em Proc. Japan Acad. Ser. A Math. Sci.}, 57(8):387--392, 1981.

\bibitem[FBZ04]{FZvi}
Edward Frenkel and David Ben-Zvi.
\newblock {\em Vertex algebras and algebraic curves}, volume~88 of {\em
  Mathematical Surveys and Monographs}.
\newblock American Mathematical Society, Providence, RI, second edition, 2004.

\bibitem[FFR91]{Triality}
Alex~J. Feingold, Igor~B. Frenkel, and John F.~X. Ries.
\newblock {\em Spinor construction of vertex operator algebras, triality, and
  {$E^{(1)}_8$}}, volume 121 of {\em Contemporary Mathematics}.
\newblock American Mathematical Society, Providence, RI, 1991.

\bibitem[FHL93]{FHL}
Igor Frenkel, Yi-Zhi Huang, and James Lepowsky.
\newblock On axiomatic approaches to vertex operator algebras and modules.
\newblock {\em Mem. Amer. Math. Soc.}, 104(494):viii+64, 1993.

\bibitem[FLM88]{FLM}
Igor Frenkel, James Lepowsky, and Arne Meurman.
\newblock {\em Vertex operator algebras and the {M}onster}, volume 134 of {\em
  Pure and Applied Mathematics}.
\newblock Academic Press Inc., Boston, MA, 1988.

\bibitem[FR97]{FR}
E.~{Frenkel} and N.~{Reshetikhin}.
\newblock {Towards Deformed Chiral Algebras}.
\newblock In {\em Proceedings of the Quantum Group Symposium at the XXIth
  International Colloquium on Group Theoretical Methods in Physics, Goslar
  1996}, pages 6023--+, 1997.

\bibitem[Fre81]{Frenkel-BF}
Igor~B. Frenkel.
\newblock Two constructions of affine {L}ie algebra representations and
  boson-fermion correspondence in quantum field theory.
\newblock {\em J. Funct. Anal.}, 44(3):259--327, 1981.

\bibitem[Hua98]{MR99m:81001}
Kerson Huang.
\newblock {\em Quantum field theory}.
\newblock John Wiley \& Sons Inc., New York, 1998.
\newblock From operators to path integrals.

\bibitem[IK11]{Iohara}
Kenji Iohara and Yoshiyuki Koga.
\newblock {\em Representation theory of the {V}irasoro algebra}.
\newblock Springer Monographs in Mathematics. Springer-Verlag London, Ltd.,
  London, 2011.

\bibitem[Kac90]{Kac-Lie}
Victor Kac.
\newblock {\em Infinite-dimensional {L}ie algebras}.
\newblock Cambridge University Press, Cambridge, third edition, 1990.

\bibitem[Kac98]{Kac}
Victor Kac.
\newblock {\em Vertex algebras for beginners}, volume~10 of {\em University
  Lecture Series}.
\newblock American Mathematical Society, Providence, RI, second edition, 1998.

\bibitem[KN87]{Krich-Nov1}
Igor~Moiseevich Krichever and S.~P. Novikov.
\newblock Algebras of {V}irasoro type, {R}iemann surfaces and strings in
  {M}inkowski space.
\newblock {\em Funktsional. Anal. i Prilozhen.}, 21(4):47--61, 96, 1987.

\bibitem[KN89]{Krich-Nov2}
Igor~Moiseevich Krichever and S.~P. Novikov.
\newblock Algebras of {V}irasoro type, the energy-momentum tensor, and operator
  expansions on {R}iemann surfaces.
\newblock {\em Funktsional. Anal. i Prilozhen.}, 23(1):24--40, 1989.

\bibitem[KR87]{KacRaina}
V.~G. Kac and A.~K. Raina.
\newblock {\em Bombay lectures on highest weight representations of
  infinite-dimensional {L}ie algebras}, volume~2 of {\em Advanced Series in
  Mathematical Physics}.
\newblock World Scientific Publishing Co. Inc., Teaneck, NJ, 1987.

\bibitem[KW94]{Wang}
Victor Kac and Weiqiang Wang.
\newblock Vertex operator superalgebras and their representations.
\newblock In {\em Mathematical aspects of conformal and topological field
  theories and quantum groups ({S}outh {H}adley, {MA}, 1992)}, volume 175 of
  {\em Contemp. Math.}, pages 161--191. Amer. Math. Soc., Providence, RI, 1994.

\bibitem[KWY98]{WangKac}
Victor~G. Kac, Weiqiang Wang, and Catherine~H. Yan.
\newblock Quasifinite representations of classical {L}ie subalgebras of {$
  W_{1+\infty}$}.
\newblock {\em Adv. Math.}, 139(1):56--140, 1998.

\bibitem[LL04]{LiLep}
James Lepowsky and Haisheng Li.
\newblock {\em Introduction to vertex operator algebras and their
  representations}, volume 227 of {\em Progress in Mathematics}.
\newblock Birkh\"auser Boston Inc., Boston, MA, 2004.

\bibitem[RT13]{Rehren}
Karl-Henning Rehren and Gennaro Tedesco.
\newblock Multilocal fermionization.
\newblock {\em Letters in Mathematical Physics}, 103: 19--36,  2013.


\bibitem[Sch94]{Schl}
Martin Schlichenmaier.
\newblock Differential operator algebras on compact {R}iemann surfaces.
\newblock In {\em Generalized symmetries in physics ({C}lausthal, 1993)}, pages
  425--434. World Sci. Publ., River Edge, NJ, 1994.

\bibitem[Wan99a]{WangDual}
Weiqiang Wang.
\newblock Dual pairs and infinite dimensional {L}ie algebras.
\newblock In {\em Recent developments in quantum affine algebras and related
  topics ({R}aleigh, {NC}, 1998)}, volume 248 of {\em Contemp. Math.}, pages
  453--469. Amer. Math. Soc., Providence, RI, 1999.

\bibitem[Wan99b]{WangDuality}
Weiqiang Wang.
\newblock Duality in infinite-dimensional {F}ock representations.
\newblock {\em Commun. Contemp. Math.}, 1(2):155--199, 1999.

\end{thebibliography}
\end{document}